\documentclass[11pt]{article}

\usepackage[numbers]{natbib}
\usepackage[left=1in,bottom=1in,right=1in,top=1in]{geometry}

\usepackage{authblk}
\usepackage{amsmath, amssymb,multirow,hyperref}
\usepackage{amsthm} 
\usepackage{mathtools, dsfont, stackrel, booktabs, bbm}
\usepackage{float}
\usepackage{capt-of}
\usepackage{url,verbatim}
\usepackage{enumitem}
\usepackage{amsmath,amssymb}
\usepackage{multirow}
\usepackage{xargs}
\usepackage{amsfonts}
\usepackage{color}
\usepackage{subfig}
\DeclareMathOperator*{\argmin}{arg\,min}
\DeclareMathOperator*{\argmax}{arg\,max}
\DeclareMathOperator{\sign}{sign}

\newtheorem{theorem}{Theorem}

\newtheorem{lemma}{Lemma}
\newtheorem{proposition}{Proposition}

\newtheorem{remark}{Remark}

\begin{document}
	\title{Sparse Regression at Scale: Branch-and-Bound rooted in First-Order Optimization}
	
	\date{April, 2021}
	
	\author{Hussein Hazimeh\thanks{MIT Operations Research Center. Email: {\texttt{hazimeh@mit.edu}}} }
	\author{
		Rahul Mazumder\thanks{MIT Sloan School of Management, Operations Research Center and MIT Center for Statistics. Email: {\texttt{rahulmaz@mit.edu}} }}
	\author{Ali Saab\thanks{AJO. Email: {\texttt{saab@mit.edu}}} }
	\affil{Massachusetts Institute of Technology}
	
	\maketitle

	\begin{abstract}
		We consider the least squares regression problem, penalized with a combination of the  $\ell_{0}$ and squared $\ell_{2}$ penalty functions (a.k.a. $\ell_0 \ell_2$ regularization). 
		Recent work shows that the resulting estimators are of key importance  in many high-dimensional statistical settings. 
		However, exact computation of these estimators remains a major challenge. 
		Indeed, modern exact methods, based on mixed integer programming (MIP), face difficulties when the number of features $p \sim 10^4$.
		In this work, we present a new exact MIP framework for $\ell_0\ell_2$-regularized regression that can scale to $p \sim 10^7$, achieving speedups of at least $5000$x, compared to state-of-the-art exact methods. Unlike recent work, which relies on modern commercial MIP solvers, we design a specialized nonlinear branch-and-bound (BnB) framework, by critically exploiting the problem structure. A key distinguishing component in our framework lies in efficiently solving the node relaxations using a specialized first-order method, based on coordinate descent (CD). Our CD-based method effectively leverages information across the BnB nodes, through using warm starts, active sets, and gradient screening. In addition, we design a novel method for obtaining dual bounds from primal CD solutions, which certifiably works in high dimensions. Experiments on synthetic and real high-dimensional datasets demonstrate that our framework is not only significantly faster than the state of the art, but can also deliver certifiably optimal solutions to statistically challenging instances that cannot be handled with existing methods. We open source the implementation through our toolkit L0BnB.
	\end{abstract}
	

	
	\newpage
	\section{Introduction}
	
	We consider the sparse linear regression problem with additional $\ell_2$ regularization \cite{bertsimas2017sparse,fastsubset}. A natural way to impose sparsity in this context is through controlling the $\ell_0$ (pseudo) norm of the estimator, which counts the number of nonzero entries. More concretely, let $X \in \mathbb{R}^{n \times p}$ be the data matrix, with $n$ samples and $p$ features, and $y \in \mathbb{R}^{n}$ be the response vector. We focus on the least squares problem with a combination of $\ell_0$ and $\ell_2$  regularization:
	\begin{align}
	\label{eq:main}
	\min_{\beta \in \mathbb{R}^p} \frac{1}{2} \| y - X \beta \|^{2}_2  + \lambda_0 \| \beta \|_0+ \lambda_2 \|\beta\|^{2}_2,
	\end{align}
	where $\|\beta \|_0$ is defined as the number of nonzero entries in the regression coefficients $\beta \in \mathbb{R}^p$, and 
	$\|\beta\|^{2}_2$ is the squared $\ell_{2}$-norm of $\beta$ (also referred to as ridge regularization).
	The regularization parameters $\lambda_0$ and $\lambda_2$ are assumed to be specified by the practitioner\footnote{The choice of $(\lambda_0,\lambda_2)$ depends upon the particular application and/or dataset. We aim to compute solutions to~\eqref{eq:main} for a family of $(\lambda_0,\lambda_2)$-values.}.
	We note that the presence of ridge regularization (i.e., $\lambda_{2}>0$) can be important from a statistical viewpoint---see for example~\cite{hastie2017extended,mazumder2017subset,fastsubset} for further discussions on this matter.
	Statistical properties of $\ell_0$-based estimators have been extensively studied in the statistics literature~\cite{greenshtein2006best,raskutti2011minimax,zhang2014lower,david2017high,wainwright2009information,mazumder2017subset}.
	{Specifically, under suitable assumptions on the underlying data and appropriate choices of tuning parameters, global solutions of~\eqref{eq:main} have optimal support recovery properties~\cite{wainwright2009information,fletcher2009necessary};
		and optimal prediction error bounds that do not depend upon $X$~\cite{raskutti2011minimax}. Appealing prediction error bounds available from optimal solutions 
		of~\eqref{eq:main} for the low-signal regime are discussed in~\cite{mazumder2017subset}. Such strong guarantees are generally not satisfied by heuristic solutions to~\eqref{eq:main}---see \cite{zhang2014lower} for theoretical support and~\cite{fastsubset} for numerical evidence. 
		From a practical perspective, the ability to certify the optimality of solutions (e.g., via dual bounds) is important and can engender trust in mission critical applications such as healthcare.
	}
	
	
	Despite its appeal, Problem~\eqref{eq:main} is labeled as NP-Hard~\cite{natarajan1995sparse} and poses computational challenges. Recently, there has been exciting work in developing Mixed Integer Programming (MIP)-based approaches to solve~\eqref{eq:main}, e.g., {\cite{cozad14,bestsubset,miyashiro2015subset,hongbo,bertsimas2017sparse,fastsubset,bertsimas2019sparse,xie2020scalable,atamturk20a}.} Specifically, \cite{bestsubset} demonstrated that off-the-shelf MIP solvers can handle problem instances for $p$ up to a thousand. Larger instances can be handled when $\lambda_2$ is sufficiently large and the feature correlations are low---for example, see the approach of~\cite{bertsimas2017sparse}; and the method in~\cite{dedieu2020learning} for the classification variant of~\eqref{eq:main}. 
	The approaches of~\cite{bertsimas2017sparse,dedieu2020learning} rely on commercial MIP solvers such as Gurobi. While these state-of-the-art global optimization approaches show very promising results, they are still relatively slow for practical usage~\cite{hastie2017extended,fastsubset}. 
	For example, our experiments show that these methods cannot terminate in two hours for typical instances with $p \sim 10^4$. On the other hand, the fast Lasso solvers, e.g., \texttt{glmnet} \cite{glmnet}, and local optimization methods for \eqref{eq:main}, such as \texttt{L0Learn} \cite{fastsubset}, can handle much larger instances, and they typically terminate in the order of milliseconds to seconds.


	
	Our goal in this paper is to advance the computational methodology for the global optimization of Problem~\eqref{eq:main}. In particular, we aim to (i) reduce the run time for solving problem instances with $p \sim 10^4$ from hours to seconds, and (ii) scale to larger problem instances with $p \sim 10^7$ in reasonable times (order of minutes to hours).  To this end, we propose a specialized nonlinear branch-and-bound (BnB) framework that does not rely on commercial MIP solvers. We employ a first-order method, which carefully exploits the problem structure, to solve the node relaxations. This makes our approach quite different from prior work on global optimization for Problem \eqref{eq:main}, which rely on commercial MIP solvers, e.g., Gurobi and CPLEX. These MIP solvers are also based on a BnB framework, but they are equipped with general-purpose relaxation solvers and heuristics that do not take into account the specific structure in Problem \eqref{eq:main}.
	
	Our BnB solves a mixed integer second order cone program (MISOCP)  {that is based on a perspective reformulation \cite{frangioni2006perspective,akturk2009strong,gunluk2010perspective} of Problem \eqref{eq:main}.} The algorithm exploits the sparsity structure in the problem during different stages: when solving node relaxations, branching, and obtaining upper bounds. The continuous node relaxations that appear in our BnB have not been studied at depth in earlier work. A main contribution of our work is to show that these relaxations, which involve seemingly complicated linear and conic constraints, can be efficiently handled using a primal coordinate descent (CD)-based algorithm. Indeed, this represents a radical change from the primal-dual relaxation solvers commonly used in state-of-the-art MIP solvers \cite{belotti2013mixed}. Our choice of CD is motivated by its ability to effectively share information across the BnB nodes (such as warm starts), and more generally by its high scalability in the context of sparse learning, e.g., see  \cite{glmnet,sparsenet,fastsubset}. Along with CD, we propose additional strategies, namely, active set updates and gradient screening, which reduce the coordinate update complexity by exploiting the information shared across the BnB tree.  
	
	Although our CD-based algorithm for solving BnB node relaxations is highly scalable, it only generates primal solutions. However, dual bounds are required for search space pruning in BnB. Thus, we propose a novel method to efficiently generate dual bounds from the primal solutions. We analyze these dual bounds and prove that their tightness is not affected by the number of features $p$, but rather by the number of nonzeros in the primal solution. This result serves as a theoretical justification for why our CD-based algorithm can lead to tight dual bounds in high dimensions. 
	

	
	\smallskip
	
	\noindent\textbf{Contributions and Structure: }
	We summarize our key contributions below.
	\begin{itemize}
		
		\item[$\bullet$] {We formulate Problem \eqref{eq:main} as a MISOCP, based on a perspective formulation. We provide a new analysis of the relaxation tightness, which identifies parameter ranges for which the perspective formulation can outperform popular formulations (see Section \ref{sec:formulations}). }
		
		\item[$\bullet$] To solve the MISOCP, we design a specialized nonlinear BnB, with the following main contributions (see Section \ref{sec:bnb}):
		\begin{itemize}
			\item We show that the node relaxations, which involve linear and conic constraints, can be reformulated as a least squares problem with a non-differentiable but separable penalty. To solve the latter reformulation, we develop a primal CD algorithm, along with active set updates and gradient screening 
			that use information shared across the BnB tree to reduce the coordinate update cost. 
			\item We develop a new efficient method for obtaining dual bounds from the primal solutions. We analyze these dual bounds and show that their tightness depends on the  sparsity level rather than $p$.
			\item We introduce efficient methods that exploit sparsity when selecting branching variables and obtaining incumbents\footnote{An incumbent, in the context of BnB, refers to the best integral solution found so far.}.
		\end{itemize}
		\item[$\bullet$] We perform a series of experiments on high-dimensional synthetic and real  datasets, with $p$ up to $8.3 \times 10^6$. We study the effect of the regularization parameters and dataset characteristics on the run time, and perform ablation studies. The results indicate that our approach can be $5000$x faster than the state of the art in some settings, and is capable of handling difficult statistical instances which were virtually unsolvable before (See Section~\ref{sec:experiments}). We open source the implementation through our toolkit L0BnB:
		\begin{center}
			\url{https://github.com/alisaab/L0BnB}
		\end{center}
	\end{itemize}
	
	\smallskip
	
	\noindent\textbf{Related Work: } As mentioned earlier, an impressive line of recent work considers solving Problem \eqref{eq:main}, or its cardinality-constrained variant, to optimality. \cite{bestsubset} used Gurobi on a Big-M formulation, which can handle $n \sim p \sim 10^3$ in the order of minutes to hours. \cite{bertsimas2017sparse} scale the problem even further by applying outer-approximation (using Gurobi) on a boolean reformulation~\cite{pilanci2015sparse} of the problem. Their approach can handle $p \sim 10^5$ in the order of minutes when $n$ and $\lambda_2$ are sufficiently large, and  the feature correlations are sufficiently small. This outer-approximation approach has also been generalized to sparse classification in \cite{bertsimas2017sparseclass}. \cite{xie2020scalable} consider solving a perspective formulation \cite{frangioni2006perspective} of the problem directly using Gurobi and reported timings that compare well with~\cite{bertsimas2017sparse}---the largest problem instances they consider have 
	$p \sim 10^3$.  \cite{dedieu2020learning}~show that a variant of Problem \eqref{eq:main} for classification can be solved through a sequence of MIPs (solved using Gurobi), each having a small number of binary variables, as opposed to $p$ binary variables in the common approaches. Their approach can handle $n=10^3$ and $p=50,000$ in minutes if the feature correlations are sufficiently small. Our specialized BnB, on the other hand, can solve all the instances mentioned above with speed-ups that exceed $5000$x, and can scale to problems with $p \sim 10^7$. Moreover, our numerical experiments show that, unlike prior work, our BnB can handle difficult problems with relatively small $\lambda_2$ and/or high feature correlations. 
	
	In addition to the global optimization approaches discussed above, there is an interesting body of work in the broader optimization community on improved relaxations for sparse ridge regression, e.g.,  \cite{pilanci2015sparse,hongbo,atamturk2019rank,xie2020scalable}, and algorithms that locally optimize an $\ell_0$-based objective \cite{blumensath2009iterative,BeckSparsityConstrained,fastsubset}.
	
	There is also a rich literature on solving mixed integer nonlinear programs  (MINLPs) using BnB, e.g., see~\cite{lee2011mixed,belotti2013mixed}. Our approach is based on the nonlinear BnB framework~\cite{dakin1965tree}, where a nonlinear subproblem is solved at every node of the search tree. Interior point methods are a popular choice for these nonlinear subproblems, especially for MISOCPs \cite{lee2011mixed}, e.g., they are used in MOSEK \cite{andersen2003implementing} {and are also one of the supported options in CPLEX \cite{studio2013users}}. Generally, interior point based nonlinear solvers are not as effective in exploiting warm starts and sparsity as linear programming solvers \cite{belotti2013mixed}, which led to an alternative approach known as outer-approximation (OA) \cite{duran1986outer}. In OA, a sequence of relaxations, consisting of mixed integer linear programs, are solved until converging to a solution of the MINLP. State-of-the-art solvers such as BARON \cite{tawarmalani2005polyhedral} and Gurobi \cite{gurobi2020gurobi} apply OA on extended formulations---{\cite{tawarmalani2005polyhedral} laid the ground work for this approach.} There is also a line of work on specialized OA reformulations and algorithms for mixed integer conic programs (which include MISOCPs), e.g., see \cite{vielma2008lifted,lubin2016extended,vielma2017extended} and the references therein. In this paper, we pursue a different approach: making use of problem-specific structure, we show that the relaxation of the MISOCP can be effectively handled by our proposed CD-based algorithm.
	
	
	
	\smallskip
	
	\noindent\textbf{Notation and Supplementary Material: } We denote the set $\{1,2,\dots,p\}$ by $[p]$. For any set $A$, the complement is denoted by $A^c$. We let $\|\cdot\|_{q}$ denote the standard $\ell_{q}$ norm with $q \in \{0,1,2,\infty\}.$ For any vector $v \in \mathbb{R}^{k}$,  $\sign(v)\in \mathbb{R}^{k}$ refers to the vector whose $i$th component is given by $\sign(v_i) = v_i/|v_i|$ if $v_i \neq 0$ and $\sign(v_i) \in [-1,1]$ if $v_i = 0$. We denote the support of $\beta \in \mathbb{R}^p$ by $\text{Supp}(\beta)=\{i: \beta_{i} \neq 0, i \in [p]\}$. For a set $S \subseteq [p]$, use $\beta_S \in \mathbb{R}^{|S|}$ to denote the subvector of $\beta$ with indices in $S$. Similarly, $X_S$ refers to the submatrix of $X$ whose columns correspond to $S$. For a scalar $a$, we denote $[a]_+ = \max\{a, 0\}$.
	{Given a set of real numbers $\{a_i\}_{i=1}^N$ and a scalar $c$, we use $\{a_i\}_{i=1}^N \cdot c$ to denote $\{c a_i\}_{i=1}^N$}. {The proofs of all propositions, lemmas, and theorems are in the appendix.} 
	
	\section{MIP Formulations and Relaxations} \label{sec:formulations}
	{In this section, we present MIP formulations for Problem \eqref{eq:main} and study their corresponding relaxations.}
	
	\subsection{MIP Formulations}

	\noindent\textbf{The Big-M Formulation:} We assume that there is a finite scalar $M>0$ (a-priori specified) such that an optimal solution of Problem~\eqref{eq:main}, say $\beta^{*}$, satisfies: $\| \beta^{*} \|_{\infty} \leq M$. 
	This allows for modeling \eqref{eq:main} as a mixed integer quadratic program (MIQP) using the following Big-M formulation:
	\begin{equation} \label{eq:bigM}
	\begin{aligned} 
	\text{B}(M): ~~~ \min_{\beta , z} ~~&~\frac{1}{2} \| y - X \beta \|^{2}_2  + \lambda_0 \sum_{i\in [p]} z_i+ \lambda_2 \|\beta\|^{2}_2  \\
	\text{s.t.}~~&~~ -M z_i \leq \beta_i \leq { M} z_i,~ i \in [p] \\
	&~~z_i \in \{0,1\},  ~ i \in [p], 
	\end{aligned}
	\end{equation}
	where each binary variable $z_{i}$ controls whether $\beta_{i}$ is zero or not via the first constraint in \eqref{eq:bigM}---i.e., if $z_{i}=0$ then $\beta_{i}=0$.
	Such Big-M formulations are widely used in mixed integer programming and have been recently explored in multiple works on $\ell_0$ regularization, e.g.,   \cite{bestsubset,mazumder2017subset,fastsubset,xie2020scalable}. See \cite{bestsubset,xie2020scalable} for a discussion on how to estimate $M$ in practice. 
	
	\medskip
	
	\noindent\textbf{The Perspective Formulation: } {In MIP problems where bounded continuous variables are activated by indicator variables, perspective reformulations  \cite{frangioni2006perspective,akturk2009strong,gunluk2010perspective} can lead to stronger MIP relaxations and thus improve the run time of BnB algorithms. Here, we apply a  perspective reformulation to the ridge term $\| \beta \|_2^2$ in Problem \eqref{eq:bigM}. Specifically, we introduce the auxiliary continuous variables $s_i \geq 0, i \in [p]$ and rotated second order cone constraints $\beta_i^2 \leq s_i z_i$, $i \in [p]$. We then replace the term $\| \beta \|_2^2$ with $\sum_{i \in [p]} s_i$. Thus, each $s_i$ takes the place of $\beta_i^2$. This leads to the following reformulation of~\eqref{eq:bigM}:
		\begin{equation}
		\label{eq:conicbigM}
		\begin{aligned}
		\text{PR}(M): ~~~ \min_{\beta, z, s} ~~&~  \frac{1}{2} \| y - X \beta \|^{2}_2  + \lambda_0 \sum_{i \in [p]} z_i+ \lambda_2 \sum_{i \in [p]} s_i \\
		\text{s.t.}~~&~~ \beta_i^2 \leq s_i z_i, ~ i \in [p]\\
		& ~~-M z_i \leq \beta_{i} \leq M z_i, ~ i \in [p] \\
		&~~ z_i \in \{0,1\}, s_i \geq 0, ~ i \in [p].
		\end{aligned}
		\end{equation}
		Problem~\eqref{eq:conicbigM} can be expressed as a MISOCP.
		Similar to~\eqref{eq:bigM}, formulation~\eqref{eq:conicbigM} is equivalent to Problem \eqref{eq:main} (as long as $M$ is suitably chosen). Algorithms for formulation~\eqref{eq:conicbigM} will be the main focus of our paper. \\
		If we set $M=\infty$ in~\eqref{eq:conicbigM}, then the constraints 
		$\beta_{i} \in [-M z_{i}, M z_{i}], i \in [p]$ can be dropped, which makes PR($\infty$)  independent of a Big-M parameter. 
		If $\lambda_2 > 0$, then PR($\infty$) is equivalent to \eqref{eq:main}---this holds since $\beta_i^2 \leq s_i z_i$ enforces $z_i = 0 \implies \beta_i = 0$. We note that \cite{hongbo} have studied Problem \eqref{eq:conicbigM} and focused on the special case of PR($\infty$). \cite{hongbo} shows that PR($\infty$) is equivalent to the pure binary formulations considered in \cite{pilanci2015sparse,bertsimas2017sparse}. We also note that \cite{xie2020scalable} have considered a similar perspective formulation for the cardinality constrained variant of Problem \eqref{eq:main}. In Proposition~\ref{prop:v1v2} below, we  present new bounds that quantify the relaxation strengths of PR($M$), PR($\infty$), and B($M$). Moreover, we are the first to present a tailored BnB procedure for formulation \eqref{eq:conicbigM}.}

	\subsection{Relaxation of the Perspective Formulation~\eqref{eq:conicbigM}} \label{sec:convex_relaxation}
	
	{In this section, we study the \textsl{interval relaxation} of Problem \eqref{eq:conicbigM}, which is obtained by relaxing all binary $z_i$'s to the interval $[0,1]$. Specifically, we present a new compact reformulation of the interval relaxation that leads to useful insights and facilitates our algorithm development. We also discuss how this reformulation compares with the interval relaxations of B$(M)$ and $\text{PR}(\infty)$.}
	
	Theorem \ref{theorem:relaxation} shows that the interval relaxation of~\eqref{eq:conicbigM} can be reformulated purely in the $\beta$ space. This leads to 
	a regularized least squares criterion, where the regularizer involves 
	the reverse Huber~\cite{owen2007robust} penalty---a hybrid between the $\ell_1$ and $\ell_2$ (squared) penalties. The reverse Huber penalty  $\mathcal{B}: \mathbb{R} \to \mathbb{R}$, is given by:
	\begin{align} \label{eq:reverse_huber}
	\mathcal{B}(t) = 
	\begin{cases}
	|t| & |t| \leq 1 \\
	(t^2+1)/{2} & |t| \geq 1.
	\end{cases}
	\end{align}
	
	\begin{theorem}{(Reduced Relaxation)} \label{theorem:relaxation}
		Let us define the functions $\psi, \psi_{1}, \psi_{2}$ as
		\begin{align*}
		\psi(\beta_i; \lambda_0, \lambda_2, M) = 
		\begin{cases}
		\psi_1(\beta_i; \lambda_0, \lambda_2) := 2 \lambda_0 \mathcal{B}( \beta_i \sqrt{{\lambda_2}/{\lambda_0}}) &  \text{if}~\sqrt{{\lambda_0}/{\lambda_2}} \leq  M \\
		\psi_2(\beta_i; \lambda_0, \lambda_2, M) := (\lambda_0/M + \lambda_2 M) |\beta_i| &    \text{if}~\sqrt{{\lambda_0}/{\lambda_2}} >  M.
		\end{cases}
		\end{align*}
		The interval relaxation of \eqref{eq:conicbigM} is equivalent to:
		\begin{align} \label{eq:relaxation}
		\min_{\beta \in \mathbb{R}^p} ~ F(\beta):= \frac{1}{2} \| y - X \beta \|_2^2 + \sum_{i \in [p]} \psi(\beta_i; \lambda_0, \lambda_2, M) ~~ \text{s.t.} ~~ \|\beta\|_{\infty} \leq M,
		\end{align}
		and we let $V_{\text{PR}(M)}$ denote the optimal objective value of~\eqref{eq:relaxation}.
	\end{theorem}
	The reduced formulation~\eqref{eq:relaxation} has an important role in both the subsequent analysis and the algorithmic development in Section \ref{sec:bnb}. 
	Theorem \ref{theorem:relaxation} shows that the conic and Big-M constraints in the relaxation of \eqref{eq:conicbigM} can be completely eliminated, at the expense of introducing (in the objective) the non-differentiable penalty function $\sum_{i} \psi(\beta_i; \lambda_0, \lambda_2, M)$ which is separable across the $p$ coordinates $\beta_{i}, i \in [p]$.
	Depending on the value of $\sqrt{{\lambda_0}/{\lambda_2}}$ compared to $M$, the penalty $\psi(\beta_i; \lambda_0, \lambda_2, M)$ is either the reverse Huber penalty (i.e., $\psi_1$) or the $\ell_1$ penalty (i.e., $\psi_2$), both of which are sparsity-inducing. {In Figure \ref{fig:penalties} (left panel), we plot $\psi_1(\beta; \lambda_0, \lambda_2)$ for $\lambda_2 = 1$ at different values of $\lambda_0$. In Figure 1 (right panel), we plot $\psi(\beta; \lambda_0, \lambda_2, M)$ at $\lambda_0 = \lambda_2 = 1$ and for different values of $M$.}
	The appearance of a pure $\ell_1$ penalty in the objective is interesting in this case since the original formulation in \eqref{eq:conicbigM} has a ridge term in the objective. Informally speaking, when $\sqrt{{\lambda_0}/{\lambda_2}} > M$, the constraint $|\beta_i| \leq M z_i$ becomes active at any optimal solution, which turns the ridge term into an $\ell_1$ penalty---for further discussion on this matter, see the proof of Theorem \ref{theorem:relaxation}.
	
	\begin{figure}[h!]
		\centering
		\begin{tabular}{cc}
			\includegraphics[width=0.48\textwidth,trim=4mm .5mm 10mm 10mm, clip]{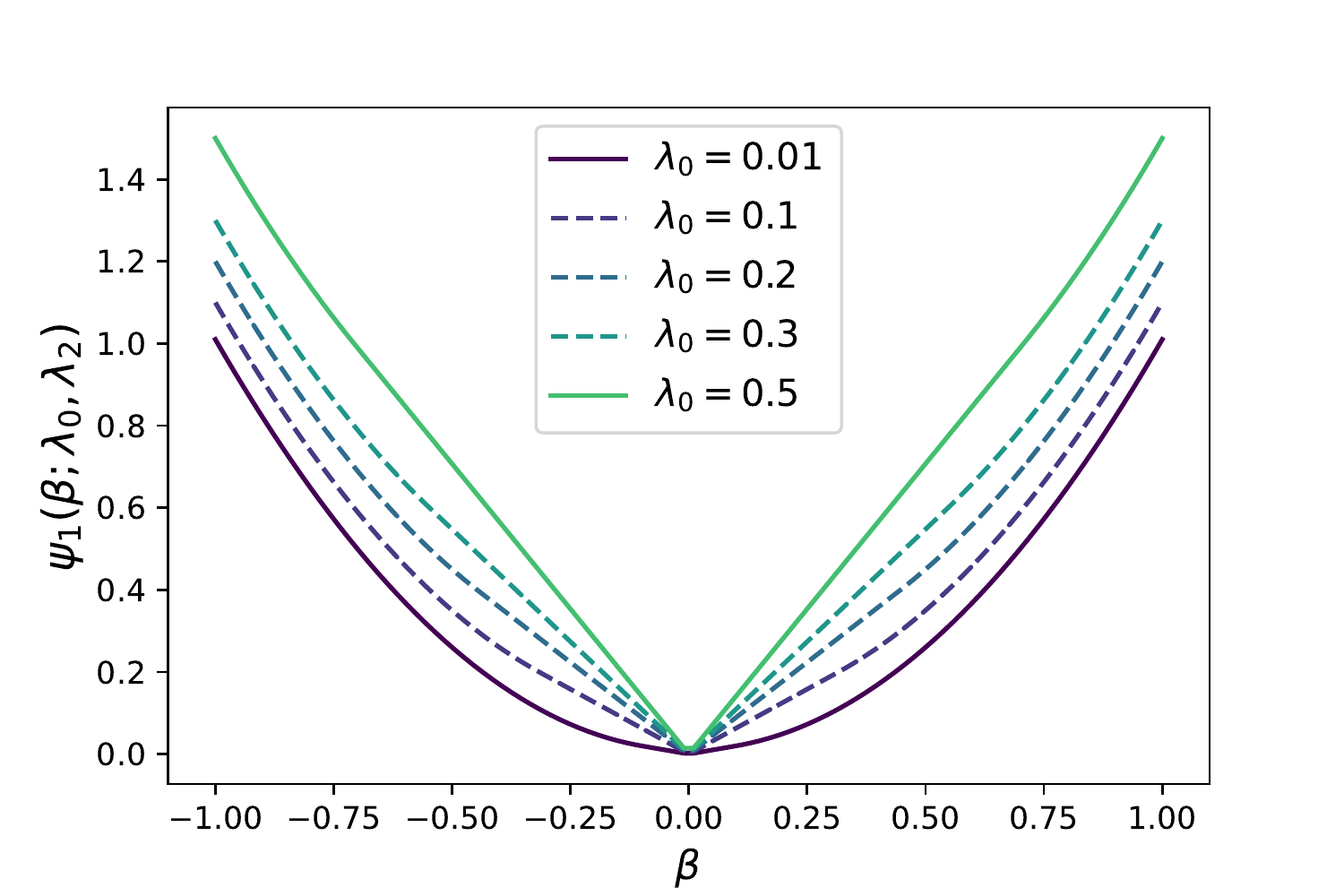}&
			\includegraphics[width=0.48\textwidth,trim=4mm .5mm 10mm 10mm,clip]{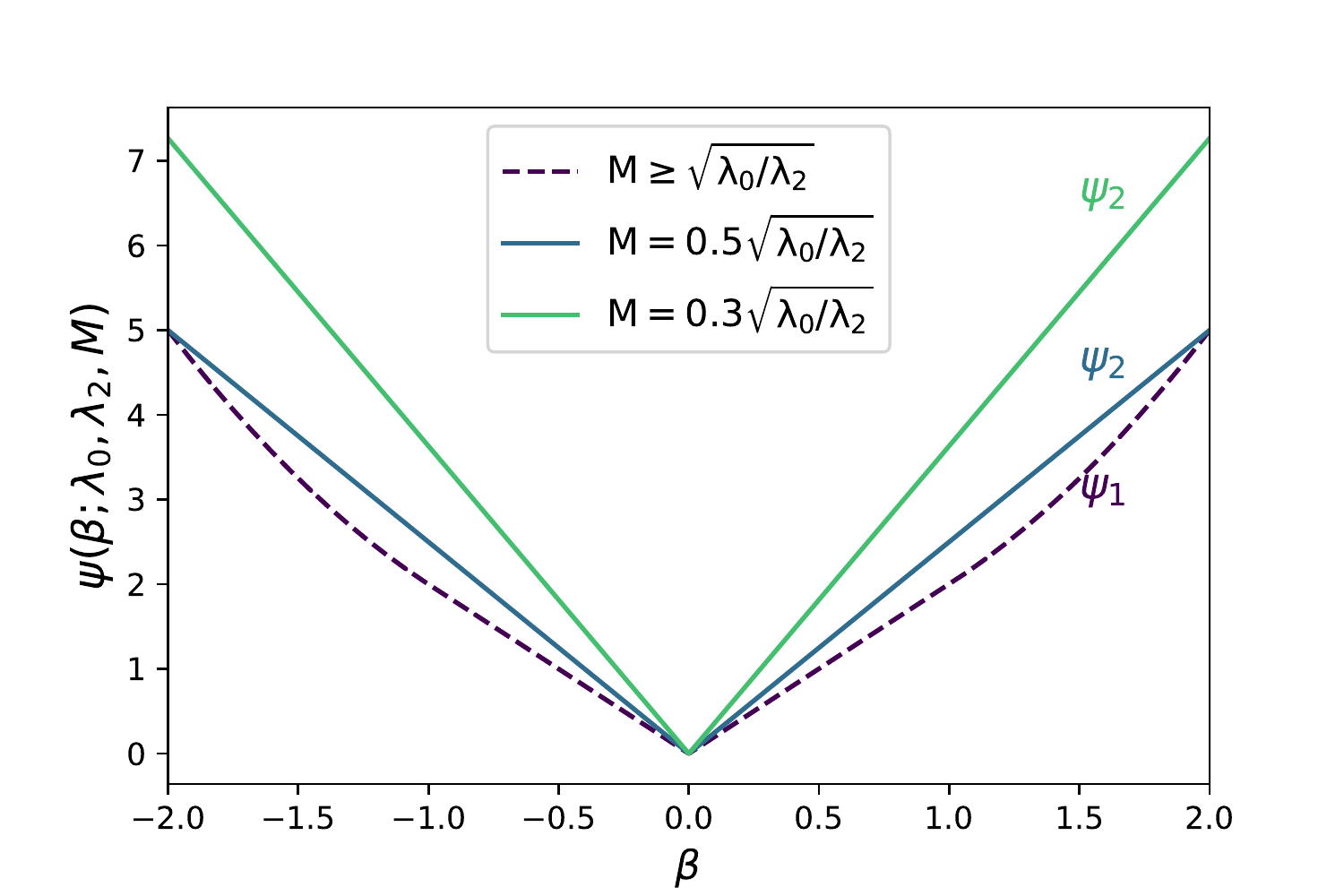}
		\end{tabular}
		\caption{[Left]: Plot of $\psi_1(\beta; \lambda_0, 1)$ for different values of $\lambda_0$. [Right]: Plot of $\psi(\beta; 1, 1, M)$ for different values of $M$.}%
		\label{fig:penalties}%
	\end{figure}

	
	Next, we analyze the tightness of the perspective relaxation in~\eqref{eq:relaxation}.
	
	\smallskip
	
	\noindent\textbf{Tightness of the Perspective Relaxation~\eqref{eq:relaxation}:}
	\cite{hongbo} has shown that the interval relaxation of PR($\infty$) can be written as: 
	\begin{align} \label{eq:conic_relaxation}
	V_{\text{PR}(\infty)} = \min_{\beta \in \mathbb{R}^p} G(\beta):= \frac{1}{2} \| y - X \beta \|_2^2 + \sum_{i\in [p]} \psi_1(\beta_i; \lambda_0, \lambda_2),
	\end{align}
	where $\psi_1(\beta_i; \lambda_0, \lambda_2)$ is defined in Theorem \ref{theorem:relaxation}. {Note that  \eqref{eq:conic_relaxation} can also be obtained from Theorem \ref{theorem:relaxation} (with $M=\infty$). For $\sqrt{{\lambda_0}/{\lambda_2}} \leq M$, the relaxation of PR($M$) in Theorem~\ref{theorem:relaxation} matches that in \eqref{eq:conic_relaxation}, but with the  additional box constraint $\| \beta \|_{\infty} \leq M$. When $M$ is large, this box constraint becomes inactive, making the  interval relaxations of PR($M$) and PR($\infty$) equivalent. However, when $\sqrt{{\lambda_0}/{\lambda_2}} > M$, the interval relaxation of PR($M$) can have a (strictly) larger objective  than that of  both B($M$) and PR($\infty$), as we show in Proposition \ref{prop:v1v2}.}
	\begin{proposition} \label{prop:v1v2}
		Let $V_{\text{B}(M)}$ denote the optimal objective of the interval relaxation of $\text{B}(M)$. Let $\beta^{*}$ be an optimal solution to \eqref{eq:relaxation}, and define the function $h(\lambda_0, \lambda_2, M) = {\lambda_0}/{M} + \lambda_2 M - 2 \sqrt{\lambda_0 \lambda_2}$. Then, the following holds for $\sqrt{{\lambda_0}/{\lambda_2}} > M$:
		\begin{align} 
		& V_{\text{PR}(M)} \geq V_{\text{B}(M)} + \lambda_2 (M \| \beta^{*} \|_1 - \|  \beta^{*} \|_2^2) \label{eq:v14}\\
		& V_{\text{PR}(M)} \geq V_{\text{PR}(\infty)} + h(\lambda_0, \lambda_2, M) \| \beta^{*} \|_1. \label{eq:v123}
		\end{align}
	\end{proposition}
	We make a couple of remarks. For bound \eqref{eq:v14}, we always have $(M \| \beta^{*} \|_1 - \|  \beta^{*} \|_2^2) \geq 0$ since $\| \beta^{*} \|_{\infty} \leq M$. If there is an $i \in [p]$ with $0 < |\beta_i^{*}| < M $, then $(M \| \beta^{*} \|_1 - \|  \beta^{*} \|_2^2) > 0$, and consequently $V_{\text{PR}(M)} > V_{\text{B}(M)}$ (as long as $\lambda_2 > 0$). In bound \eqref{eq:v123}, for  $\sqrt{{\lambda_0}/{\lambda_2}} > M$, $h(\lambda_0, \lambda_2, M)$ is strictly positive and monotonically decreasing in $M$, which implies that $V_{\text{PR}(M)} > V_{\text{PR}(\infty)}$ (as long as $\beta^{*} \neq 0$). {In Section \ref{sec:timing_comparison}, our experiments empirically validate Proposition~\ref{prop:v1v2}: using PR$(M)$ with a sufficiently tight (but valid) $M$ can speed up the same BnB solver by more than $90$x compared to PR$(\infty)$.} 
	
	Note that Proposition \ref{prop:v1v2} does not directly compare  $V_{\text{PR}(\infty)}$ with $V_{\text{B}(M)}$. In Proposition \ref{prop:relax_obj}, we establish a new result which compares $V_{\text{PR}(\infty)}$ with $V_{\text{B}(M)}$. Before we present Proposition~\ref{prop:relax_obj}, we introduce some notation. Let $\mathcal{S}(\lambda_2)$ be the set of optimal solutions 
	(in $\beta$) of the interval relaxation of PR($\infty$); and define 
	\begin{align} \label{eq:lambda2star}
	\mathcal{L}(M) := \{ \lambda_2 > 0 \ | \ \exists \beta \in \mathcal{S}(\lambda_2) \text{ s.t. } \| \beta \|_{\infty} \leq M \}.
	\end{align}
	\begin{proposition} \label{prop:relax_obj}
		Let $\mathcal{L}(M)$ be as defined in \eqref{eq:lambda2star}. Then, the following holds:
		\begin{align}
		& V_{\text{B}(M)} \geq V_{\text{PR}(\infty)}, \text{ if } M \leq \tfrac12\sqrt{{\lambda_0}/{\lambda_2}} \label{eq:relax_obj_1} \\
		& V_{\text{B}(M)} \leq V_{\text{PR}(\infty)}, \text{ if } M \geq \sqrt{{\lambda_0}/{\lambda_2}} \text{ and } \lambda_2 \in \mathcal{L}(M). \label{eq:relax_obj_3}
		\end{align}
	\end{proposition}
	Proposition \ref{prop:relax_obj} implies that if $\lambda_2$ is relatively small (with other parameters remaining fixed), then the relaxation of the Big-M formulation  (with objective $V_{\text{B}(M)}$) will have a higher objective than the relaxation of PR$(\infty)$. On the other hand, if $\lambda_2$ is sufficiently large, then the relaxation of PR($\infty$) will have a higher objective.
	
	{We note that the result of Proposition \ref{prop:relax_obj} applies for any $M \geq 0$, even if it is mis-specified. However, in the case of mis-specification, $V_{\text{B}(M)}$ (and also $V_{\text{PR}(M)}$) may no longer correspond to a valid relaxation of Problem \eqref{eq:main}. In contrast, $V_{\text{PR}(\infty)}$ is always a valid relaxation of~\eqref{eq:main}.}

	\section{A Specialized Branch-and-Bound (BnB) Framework} \label{sec:bnb}
	In this section, we develop a specialized nonlinear BnB framework for solving the perspective formulation in \eqref{eq:conicbigM}. First, we briefly recall the high-level mechanism behind nonlinear BnB, in the context of our problem.
	
	\smallskip
	
	\noindent\textbf{Nonlinear BnB at a Glance: }
	The algorithm starts by solving the (nonlinear) interval relaxation of \eqref{eq:conicbigM}, i.e., the root node. Then, it selects a branching variable, say variable $j \in [p]$, and creates two new nodes (optimization subproblems): one node with $z_j = 0$ and another with $z_j = 1$, where all the other $z_i$'s are relaxed to the interval $[0,1]$. For every unvisited node, the algorithm proceeds recursively, i.e., by solving an optimization subproblem at the current node and then branching on a new variable to create two new nodes. This leads to a search tree with nodes corresponding to optimization subproblems and edges representing branching decisions.
	
	To reduce the size of the search tree, BnB prunes a node (i.e., does not branch on it) in either one of the following situations: (i) an optimal solution to the relaxation at the current node has an integral $z$ or (ii) the objective of the current relaxation exceeds the best available upper bound 
	on~\eqref{eq:conicbigM}. {In case (ii), the relaxation need not be solved exactly: lower bounds on the relaxation's objective (a.k.a. dual bounds) can be used for pruning. However, in case (i), if the dual bound does not exceed the best upper bound, the node should be solved to optimality in order to ensure correctness of the  pruning decision\footnote{In practice, we solve the relaxation problem at the node to optimality by ensuring that the relative difference between the primal and dual bounds is less than a small user-defined numerical tolerance.}}.
	
	
	{As BnB explores more nodes, its (global) lower bound is guaranteed to converge to the optimal objective of \eqref{eq:conicbigM}. In practice, we can terminate the algorithm early, if the gap between the best lower and upper bounds is below a pre-specified user-defined threshold.}
	
	\smallskip
	
	\noindent\textbf{Overview of our Strategies: }
	The discussion above outlines how nonlinear BnB operates in general. The specific strategies used such as solving the relaxations, passing information across the nodes, and selecting branching variables, can have a key impact on scalability. In the rest of this section, we will give a detailed account of the strategies used in our BnB. We first provide an overview of these strategies:
	\begin{itemize}
		\item \textbf{A Primal Relaxation Solver}: Unlike state-of-the-art approaches for nonlinear BnB, which employ primal-dual interior point solvers for node relaxations~\cite{belotti2013mixed}, we rely solely on a primal method which consists of a highly scalable CD-based algorithm. The algorithm solves the node relaxations of \eqref{eq:conicbigM} in the $\beta$-space as opposed to the extended $(\beta, s, z)$ space---these relaxations are variants of the reduced relaxation introduced in \eqref{eq:relaxation}. The algorithm heavily shares and exploits warm starts, active sets, and information on the gradients, across the BnB tree. This will be developed in Section \ref{sec:CD}.
		\item \textbf{Dual Bounds}: Dual bounds on the objective of node relaxations are required by BnB for search space pruning, yet our relaxation solver works in the primal space for scalability considerations. We develop a new efficient method for obtaining dual bounds from the primal solutions. We provide an analysis of this method and show that the tightness of the dual bounds depends on the sparsity level and \emph{not} on the number of features $p$. See Section \ref{sec:dualbds}.
		\item \textbf{Branching and Incumbents}: We present an efficient variant of strong branching, which leverages the solutions and active sets of previous node relaxations to make optimization tractable. Moreover, we employ several efficient heuristics to obtain incumbents. See Section \ref{sec:branching}.
	\end{itemize}

	For simplicity of exposition, in the remainder of Section \ref{sec:bnb}, we assume that the columns of $X$ and $y$ have unit $\ell_2$ norm.
	\subsection{Primal Relaxation Solver: Active-set Coordinate Descent} \label{sec:CD}
	To simplify the presentation, we will focus on solving the root relaxation. To this end, we solve the reduced formulation in the $\beta$-space~\eqref{eq:relaxation}. 
	We operate on the reduced formulation compared to the interval relaxation of~\eqref{eq:conicbigM} in the extended $(\beta, s, z)$ space due to computational reasons. 
	{After solving \eqref{eq:relaxation}, we use the resulting solution, say $\beta^{*}$, to construct a corresponding solution $(\beta^{*}, s^{*}, z^{*})$ to the interval relaxation of \eqref{eq:conicbigM}---see the proof of Theorem \ref{theorem:relaxation} for how to obtain $(\beta^{*}, s^{*}, z^{*})$ from $\beta^{*}$. We then use $z^{*}$ for branching.} The rest of the nodes in BnB involve fixing some binary variables in \eqref{eq:conicbigM} to $0$ or $1$, so their subproblems can be obtained by minor modifications to the root relaxation. {For completeness, in Appendix \ref{appendix:node_subproblems}, we discuss how to formulate and solve these node subproblems.}
	
	Problem~\eqref{eq:relaxation} is of the composite form~\cite{NesterovComposite}: the objective is the sum of a smooth loss function and a non-smooth but separable penalty. In addition, the feasible set, consisting of the constraints $|\beta_i| \leq M$, $i \in [p]$ is separable across the coordinates. This makes Problem \eqref{eq:relaxation} amenable to cyclic CD~\cite{Tseng2001}.
	{To our knowledge, the use of cyclic CD for Problem~\eqref{eq:relaxation} is novel. We also emphasize that a direct application of cyclic CD to Problem \eqref{eq:relaxation} will face scalability issues, and more importantly, {it does not readily deliver dual bounds}. The additional strategies we develop later in this section are essential for both achieving scalability and obtaining (provably) high quality dual bounds.}
	
	Cyclic CD visits the coordinates according to a fixed ordering, updating one coordinate at a time, as detailed in Algorithm 1.
	
	\begin{itemize}
		\item[] \textbf{Algorithm 1: Cyclic CD for Relaxation \eqref{eq:relaxation}}
		\item \textbf{Input:} Initialization $\hat{\beta}$
		\item \textbf{While} not converged:
		\begin{itemize}[noitemsep]
			\item \textbf{For} $i \in [p]: $
			\begin{align} \label{eq:univariate}
			\hat{\beta}_i \gets \argmin_{\beta_i \in \mathbb{R}} F(\hat{\beta}_1, \dots, \beta_i, \dots, \hat{\beta}_p) ~~ \text{ s.t. }~~ |\beta_i| \leq M.
			\end{align}
		\end{itemize}
	\end{itemize}
	{Every limit point of Algorithm 1 is a global minimizer of \eqref{eq:relaxation} \cite{Tseng2001}, and the algorithm has a sublinear rate of convergence \cite{hong2017iteration}.} 
	We will show that the solution of \eqref{eq:univariate} can be computed in closed-form. To this end, since the columns of $X$ have unit $\ell_{2}$-norm, we note that~\eqref{eq:univariate} is equivalent to:
	\begin{align} \label{eq:proximal}
	\min_{\beta_i \in \mathbb{R}}~~\frac{1}{2} ( \beta_i - \tilde{\beta}_i )^2 + \psi(\beta_i; \lambda_0, \lambda_2, M)  ~~~ \text{ s.t. }~~~ |\beta_i| \leq M,
	\end{align}
	where $\tilde{\beta}_i := \langle y - \sum_{j \neq i} X_j \hat{\beta}_j, X_i \rangle$. 
	Given non-negative scalar parameters $a$ and $m$, we define the \textsl{boxed soft-thresholding operator} $T: \mathbb{R} \to \mathbb{R}$ as
	$$
	T(t; a, m) := \begin{cases}
	0 & \text{ if } |t| \leq a \\
	(|t| - a) \sign(t) & \text{ if } a < |t| \leq a + m\\
	m \sign(t) & \text{ otherwise. }
	\end{cases}
	$$
	For $\sqrt{{\lambda_0}/{\lambda_2}} \leq  M$, the solution of \eqref{eq:proximal} is given by:
	\begin{align} \label{eq:case1sol}
	\hat{\beta}_i = \begin{cases}
	T(\tilde{\beta}_i; 2 \sqrt{\lambda_0 \lambda_2}, M) & \text{ if } |\tilde{\beta}_i| \leq 2 \sqrt{\lambda_0 \lambda_2} + \sqrt{{\lambda_0}/{\lambda_2}} \\
	T(\tilde{\beta}_i(1 + 2 \lambda_2)^{-1}; 0, M) & \text{ otherwise }
	\end{cases}
	\end{align}
	and for $\sqrt{{\lambda_0}/{\lambda_2}} >  M$, the solution of \eqref{eq:proximal} is:
	\begin{align} \label{eq:case2sol}
	\hat{\beta}_i = T(\tilde{\beta}_i; \lambda_0/M + \lambda_2 M, M).
	\end{align}
	Thus, update \eqref{eq:univariate} in Algorithm 1 can be computed in closed-form using  expressions \eqref{eq:case1sol} and \eqref{eq:case2sol}. {See Appendix \ref{appendix:thresholding} for a derivation of the update rules}.
	
	\smallskip
	
	\noindent\textbf{Warm Starts: } {Interior-point methods used in state-of-the-art nonlinear BnB solvers cannot easily use warm starts. On the other hand, cyclic CD has proven to be very effective in exploiting warm starts, e.g., see  \cite{glmnet,fastsubset}. For every node in BnB, we use its parent's primal solution as a warm start. Recall that the parent and children nodes have the same relaxations, except that a single binary variable (used for branching) is fixed to zero or one in the children\footnote{{Suppose the parent node branches on $z_i$. After reformulating the node relaxations in the $\beta$ space (as discussed in Appendix  \ref{appendix:node_subproblems}), the relaxation of the child with $z_i = 1$ will be the same as the parent except that the penalty of coordinate $i$ in the parent, $\psi(\beta_i; \lambda_0, \lambda_2, M)$, will be replaced with $\lambda_0 + \lambda_2 \beta_i^2$. For the child with $z_i = 0$, the relaxation will be similar to the parent but with $\beta_i$ fixed to $0$. }}. Intuitively, this can make the warm start close to the optimal solution. Moreover, the supports of the optimal solutions of the parent and children nodes are typically very close. We exploit this observation by sharing information about the supports across the BnB tree, as we describe next in Section \ref{sec:activesets}.}
	
	\subsubsection{Active Sets} \label{sec:activesets}
	Update \eqref{eq:univariate} in Algorithm 1 requires $\mathcal{O}(n)$ operations, so every full cycle (across all $p$ coordinates) of the algorithm has cost of $\mathcal{O}(np)$. This becomes a major bottleneck for large $n$ or $p$, since CD can require many cycles before convergence. In practice, the majority of the variables stay zero during the course of  Algorithm 1 (assuming that the regularization parameters are chosen so that the optimal solution {of the relaxation is} sparse, and good warm starts are used). Motivated by this observation, we run Algorithm 1 restricted to a small subset of the variables $\mathcal{A} \subseteq [p]$, which we refer to as the \textsl{active set}, i.e., we solve the following restricted problem:
	\begin{align} \label{eq:restricted_problem}
	\hat{\beta} \in \argmin_{\beta \in \mathbb{R}^{p}} F(\beta) ~~~~ \text{s.t. } ~~~~ \|\beta \|_{\infty} \leq M, ~~ \beta_{\mathcal{A}^c} = 0.
	\end{align} 
	After solving \eqref{eq:restricted_problem}, we augment the active set with any variable $i \in \text{Supp}(\hat{\beta})^c$ that violates the following optimality condition:
	$$0 = \argmin_{\beta_i \in \mathbb{R}} F(\hat{\beta}_1, \dots, \beta_i, \dots, \hat{\beta}_p) ~~ \text{ s.t. } ~~ |\beta_i| \leq M.$$
	We repeat this procedure of solving the restricted problem \eqref{eq:restricted_problem} and then augmenting $\mathcal{A}$ with violating variables, until there are no more violations. The algorithm is summarized below.
	
	\begin{itemize}[leftmargin=*]
		\item[] \textbf{Algorithm 2: The Active-set Algorithm\footnote{Algorithm 2 solves the root relaxation. For other nodes, the objective function $F(\beta)$ will need to be modified to account for the fixed variables (as detailed in Appendix \ref{appendix:node_subproblems}), and all the variables fixed to zero at the current node should be excluded from the set $\mathcal{V}$ in Step 2.}}
		\item \textbf{Input:} Initial solution $\hat{\beta}$ and initial active set $\mathcal{A}$.
		\item \textbf{Repeat}:
		\begin{itemize}[leftmargin=*]
			\item Step 1: Solve \eqref{eq:restricted_problem} using Algorithm 1 to get a solution $\hat{\beta}$.
			\item Step 2: $\mathcal{V} \gets \{i \in \text{Supp}(\hat{\beta})^c \ | \ 0 \neq \argmin_{|\beta_i| \leq M} F(\hat{\beta}_1, \dots, \beta_i, \dots, \hat{\beta}_p)   \}$.
			\item Step 3: If $\mathcal{V}$ is empty \textbf{terminate}, otherwise, $\mathcal{A} \gets \mathcal{A} \cup \mathcal{V}$.
		\end{itemize}
	\end{itemize}
	
	The quality of the initial active set affects the number of iterations in Algorithm 2. For example, if $\mathcal{A}$ is a superset of the support of an optimal solution to \eqref{eq:relaxation}, then $\mathcal{V}$ in the first iteration of Algorithm 2 will be empty, and the algorithm will terminate in a single iteration. For every node in BnB, we choose the initial active set to be the same as the final active set of the parent node. This works well in practice because parent and children nodes typically have similar supports. For the root relaxation, we obtain the initial active set by choosing a small subset of features which have the highest (absolute) correlation with $y$. 
	
	In Section \ref{sec:dualbds}, we present a novel method that makes use of the updates in Algorithm~2 to obtain provably high-quality dual bounds. We note that our approach goes beyond standard active-set methods \cite{glmnet}. In particular, (i) we use active sets in the context of a BnB tree, percolating information from parents to children (as opposed to warm-start continuation across a grid of regularization parameters); and (ii) we exploit the active sets to deliver dual bounds. 
	In the next remark, we discuss how Step 1 in Algorithm 2 can be performed inexactly.
	\begin{remark} \label{remark:inexactness}
		In practice, we solve the restricted optimization problem in Step~1 of Algorithm~2 inexactly. Specifically, in Step~1, we terminate Algorithm 1 when the relative change in the objective values is below a small numerical tolerance\footnote{If upon termination, an integral $z$ is obtained, we solve the relaxation to optimality, as discussed in the introduction of  Section \ref{sec:bnb}.}. For Step~3, we ensure that $\mathcal{V}$ is (exactly) empty before terminating the algorithm, which guarantees that there are no optimality violations outside the support. Having no optimality violations outside the support will be essential to obtain the tight dual bounds discussed in Section \ref{sec:dualbds}.
	\end{remark}

	\subsubsection{Gradient Screening} \label{sec:gradient_screening}
	Algorithm 2 effectively reduces the number of coordinate updates, through restricting optimization to the active set. However, checking the optimality conditions outside the active set, i.e., performing Step 2, requires $\mathcal{O}(np)$ operations. This is a bottleneck when $p$ is large, even if a small number of such checks (or passes) are performed. To mitigate this, we present a \textsl{gradient screening} method which reduces the time complexity of these optimality checks. Our method is inspired by the gradient screening technique proposed in \cite{hazimeh2019learning} for a different problem: learning sparse hierarchical interactions via a convex optimization formulation. In the current paper, the optimality checks in Step 2 of Algorithm 2 essentially require computing a gradient of the least squares loss, in order to construct $\mathcal{V}$. Loosely speaking, gradient screening is designed to avoid computing the ``non-essential'' parts of this gradient by using previously computed quantities in the BnB tree.
	
	In the following proposition, we give an explicit way to construct the set $\mathcal{V}$ in Step 2 of Algorithm 2. 
	\begin{proposition} \label{prop:V_explicit}
		Let $\hat{\beta}$ and $\mathcal{V}$ be as defined in Algorithm 2, and define  $\hat{r} = y - X \hat{\beta}$. Then, the set $\mathcal{V}$ can be equivalently written as follows:
		\begin{align} \label{eq:V}
		\mathcal{V} = \{  i \in \text{Supp}(\hat{\beta})^c \ | \ | \langle \hat{r}, X_i \rangle| > c(\lambda_0, \lambda_2, M) \},
		\end{align}
		where $X_{i}$ denotes the $i$-th column of $X$; and
		$c(\lambda_0, \lambda_2, M) = 2 \sqrt{\lambda_0 \lambda_2}$ if $\sqrt{{\lambda_0}/{\lambda_2}} \leq M$, and $c(\lambda_0, \lambda_2, M) = (\lambda_0/M + \lambda_2 M)$ if $\sqrt{{\lambda_0}/{\lambda_2}} > M$.
	\end{proposition}
	
	By Proposition \ref{prop:V_explicit}, constructing $\mathcal{V}$ directly costs $\mathcal{O}(n (p - \| \hat{\beta} \|_0 ))$. Next, we will discuss how to compute $\mathcal{V}$ with a lower cost, by making use of previously computed quantities. Suppose we have access to a set $\hat{\mathcal{V}} \subseteq [p]$ such that $\mathcal{V} \subseteq \hat{\mathcal{V}}$. Then, we can construct $\mathcal{V}$ by restricting the checks in \eqref{eq:V} to $\hat{\mathcal{V}}$ instead of $\text{Supp}(\hat{\beta})^c$, i.e., the following holds:
	\begin{align} \label{eq:Vrestricted}
	\mathcal{V} = \{  i \in \hat{\mathcal{V}} \ | \ | \langle \hat{r}, X_i \rangle| > c(\lambda_0, \lambda_2, M) \}.
	\end{align}
	Assuming $\hat{\mathcal{V}}$ is available, the cost of computing $\mathcal{V}$ in \eqref{eq:Vrestricted} is $\mathcal{O}(n |\hat{\mathcal{V}}|)$. This cost can be significantly smaller than that of \eqref{eq:V}, if $|\hat{\mathcal{V}}|$ is sufficiently small. Next, we present a method that can obtain a relatively small $\hat{\mathcal{V}}$ in practice, thereby speeding up the computation of $\mathcal{V}$.
	
	\smallskip
	
	\noindent\textbf{Computation of $\hat{\mathcal{V}}$: }
	Proposition \ref{prop:screening} presents a method to construct a set  $\hat{\mathcal{V}}$ which satisfies {$\mathcal{V} \subseteq \hat{\mathcal{V}}$} (as discussed above), using information from a warm start $\beta^{0}$ (e.g., the solution of the relaxation from the parent node in BnB). 
	\begin{proposition} \label{prop:screening}
		Let $\hat{\beta}$ and $\mathcal{V}$ be as defined in Algorithm 2. Let $\beta^{0}$ be an arbitrary vector in $\mathbb{R}^p$. Define $\hat{r} = y - X \hat{\beta}$,  $r^{0} = y - X \beta^{0}$, and $\epsilon = \|X \beta^{0} - X \hat{\beta}\|_2$. Then, the following holds:
		\begin{align} \label{eq:V_hat}
		\mathcal{\mathcal{V}} \subseteq \hat{\mathcal{V}} := \Big \{ i \in \text{Supp}(\hat{\beta})^c \ | \ | \langle r^0, X_i \rangle | > c(\lambda_0, \lambda_2, M) -  \epsilon   \Big\}.
		\end{align}
	\end{proposition}

	The set $\hat{\mathcal{V}}$ defined in \eqref{eq:V_hat}  depends solely on $\epsilon$ and the quantities $| \langle {r}^{0}, X_i \rangle |$, $i \in \text{Supp}(\hat{\beta})^c$, which we will assume to be sorted and available along with the warm start $\beta^{0}$. This is in contrast to a direct evaluation of $\mathcal{V}$ in \eqref{eq:V}, which requires the costly computation of the terms $| \langle \hat{r}, X_i \rangle |$, $i \in \text{Supp}(\hat{\beta})^c$. Given the sorted values $| \langle {r}^{0}, X_i \rangle |$, $i \in \text{Supp}(\hat{\beta})^c$, we can identify $\hat{\mathcal{V}}$
	in \eqref{eq:V_hat} with $\mathcal{O}(\log p)$ operations, using a variant of binary search. Thus, 
	the overall cost of computing  $\mathcal{V}$ using \eqref{eq:Vrestricted} is $\mathcal{O}(n|\hat{\mathcal{V}}| + \log p)$. {We also note that in practice the inclusion ${\mathcal V} \subseteq \hat{\mathcal V}$ in \eqref{eq:V_hat} is strict.}
	
	Proposition \ref{prop:screening} tells us that the closer $X\beta^{0}$ is to $X \hat{\beta}$, the smaller $\hat{\mathcal{V}}$ is, i.e., the lower is the cost of computing $\mathcal{V}$ in \eqref{eq:Vrestricted}. Thus, for the method to be successful, we need a good warm start $\beta^{0}$. In practice, we obtain $\beta^{0}$ for the current node in BnB from its parent, and we update $\beta^{0}$ as necessary during the course of Algorithm 2 (as detailed below). 
	{Let $\epsilon_{\text{gs}} \in (0,1)$ be a pre-specified parameter for the gradient screening procedure.} The procedure, which replaces Step 2 of Algorithm 2, is defined as follows:
	\begin{enumerate}
		\item[] \textbf{Gradient Screening}
		\item If this is the root node and the first iteration of Algorithm 2, then: update $\beta^{0} \gets \hat{\beta}$, $r^{0} = y - X \beta^{0}$; compute and sort the terms $|\langle r^{0}, X_i \rangle|$, $i \in [p]$; compute $\mathcal{V}$ using \eqref{eq:V}; and skip all  the steps below.
		\item If this is the first iteration of Algorithm 2, get $\beta^{0}$ from the parent node in the BnB tree.
		\item Compute $\hat{\mathcal{V}}$ using \eqref{eq:V_hat} and then $\mathcal{V}$ using \eqref{eq:Vrestricted}\footnote{Each variable $i$ fixed to zero by BnB at the current node should be excluded when computing ${\mathcal{V}}$ and $\hat{\mathcal{V}}$.}.
		\item  If $|\hat{\mathcal{V}}| > \epsilon_{\text{gs}} p$, update $\beta^{0} \gets \hat{\beta}$, $r^{0} = y - X \beta^{0}$; re-compute and sort the terms $|\langle r^{0}, X_i \rangle|$, $i \in [p]$.
	\end{enumerate}
	
	If $X \beta^{0}$ is not a good estimate of $X \hat{\beta}$, then $\hat{\mathcal{V}}$ might become large. To avoid this issue, we update $\beta^{0}$ in Step 4 above, every time the set $\hat{\mathcal{V}}$ becomes relatively large ($|\hat{\mathcal{V}}| > \epsilon_{\text{gs}} p$). {The parameter $\epsilon_{\text{gs}}$ controls how often $\beta^{0}$ is updated. In our implementation, we use $\epsilon_{\text{gs}} = 0.05$ by default, but we note that the parameter can be generally tuned in order to improve the running time.} The updated $\beta^{0}$ will be passed to the children in the BnB tree. While Step 4 above can be costly, it is not performed often in practice as the solutions of the relaxations in the parent and children nodes are typically close.
	
	{Based on our current implementation and experiments, we see notable benefits from gradient screening when $p \geq 10^4$ (e.g., more than $2$x speedup for $p=10^6$). For smaller values of $p$, we typically observe a small additional overhead from gradient screening. 
		Moreover, we note that gradient screening will increase memory consumption, because each of the open nodes in the tree will need to store a $p$-dimensional vector (consisting of the quantities $|\langle r^{0}, X_i \rangle|$, $i \in [p]$, which are maintained by gradient screening).} 

	\subsection{Dual Bounds} \label{sec:dualbds}
	In practice, we use Algorithm 2 to obtain inexact primal solutions for relaxation \eqref{eq:relaxation}, as discussed in Remark \ref{remark:inexactness}. However, dual bounds are needed to perform search space pruning in BnB. Here, we present a new efficient method to obtain dual bounds from the primal solutions. Moreover, we prove that our method can obtain dual bounds whose tightness depends on the sparsity level of the relaxation rather than $p$. We start by introducing a Lagrangian dual of relaxation \eqref{eq:relaxation} in Theorem \ref{theorem:duals}.
	\begin{theorem} \label{theorem:duals}
		For $\sqrt{{\lambda_0}/{\lambda_2}} \leq  M $, a dual of Problem \eqref{eq:relaxation} is given by:
		\begin{align}
		\max_{ \alpha \in \mathbb{R}^n, \gamma \in {\mathbb{R}}^{p}} & ~~ h_1(\alpha, \gamma) := - \frac{1}{2} \| \alpha \|_2^2 - \alpha^T y - \sum_{i \in [p]} v(\alpha, \gamma_i), \label{eq:objdualcase1}
		\end{align}
		where $v: \mathbb{R}^{n+1} \to \mathbb{R}$ is defined as follows:
		\begin{align}\label{thm2-defn-v1}
		v(\alpha, \gamma_i) := \Big[   \frac{(\alpha^T X_i -  \gamma_i)^2}{4 \lambda_2} - \lambda_0  \Big]_{+} + M |\gamma_i|.
		\end{align}
		Otherwise, if $\sqrt{{\lambda_0}/{\lambda_2}} >  M $, a dual of Problem \eqref{eq:relaxation} is given by
		\begin{equation} \label{eq:objdualcase2}
		\begin{aligned}
		\max_{\rho \in \mathbb{R}^n, \mu \in \mathbb{R}^{p}} &~ h_2(\rho, \mu) := - \frac{1}{2} \| \rho \|_2^2 - \rho^T y - M \|\mu\|_1 \\
		\text{s.t.}~~&~~ |\rho^T X_i| - \mu_i \leq \lambda_0/M + \lambda_2 M, ~ i \in [p].
		\end{aligned}
		\end{equation}
		Let $\beta^{*}$ be an optimal solution to Problem \eqref{eq:relaxation} and define $r^{*} = y - X\beta^{*}$. The optimal dual variables for \eqref{eq:objdualcase1} are given by: 
		\begin{equation}\label{thm2-statement-pr-dual1}
		\alpha^{*} = -r^{*}~~\text{and}~~\gamma^{*}_i = \mathds{1}_{[|\beta_i^{*}| = M]}({\alpha}^{*T} X_i -  2M\lambda_2 \sign({\alpha}^{*T} X_i)), i \in [p].
		\end{equation}
		Moreover, the optimal dual variables for \eqref{eq:objdualcase2} are: 
		\begin{align}\label{thm2-statement-pr-dual2}
		\rho^{*} = -r^{*}~~~\text{and}~~~\mu^{*}_i = \mathds{1}_{[|\beta_i^{*}| = M]}(|\rho^{*T} X_i| - \lambda_0/M - \lambda_2 M), i \in [p].
		\end{align}
	\end{theorem}
	Note that strong duality holds for relaxation \eqref{eq:relaxation} since it satisfies Slater's condition \cite{bertsekas2016nonlinear}, so the optimal objective of the dual in Theorem \ref{theorem:duals} matches that of \eqref{eq:relaxation}.

	
	\smallskip
	
	\noindent\textbf{Dual Feasible Solutions: } Let $\hat{\beta}$ be an inexact primal solution obtained using Algorithm 2 and define $\hat{r} = y - X\hat{\beta}$. We discuss next how to construct a dual feasible solution
	using $\hat{\beta}$, i.e., without solving the dual in Theorem \ref{theorem:duals}. 
	
	\noindent{\bf{Case of $\sqrt{{\lambda_0}/{\lambda_2}} \leq  M$:} } 
	Here, we construct a dual solution $(\hat{\alpha}, \hat{\gamma})$ for Problem \eqref{eq:objdualcase1} as follows:
	\begin{align} \label{eq:case1dualsol}
	\hat{\alpha} = - \hat{r} ~~ \text{and} ~~  \hat{\gamma} \in \argmax_{\gamma \in \mathbb{R}^p} h_1(\hat{\alpha}, \gamma).
	\end{align}
	The choice $\hat{\alpha} = - \hat{r}$ is motivated by the optimality conditions in Theorem \ref{theorem:duals}, while $\hat{\gamma}$ maximizes the dual objective (with $\alpha$ fixed to $\hat{\alpha}$). Note that the constructed solution is (trivially) feasible since the dual in \eqref{eq:objdualcase1} is unconstrained. Since \eqref{eq:objdualcase1} is separable across the $\gamma_i$'s,  $\hat{\gamma}_i$ is equivalently the solution of $\min_{\gamma_i \in \mathbb{R}} v(\hat{\alpha}, \gamma_i)$, 
	whose corresponding solution is given by:
	\begin{align} \label{eq:case1dualsolexplicit}
	\hat{\gamma}_i =
	T(\hat{\alpha}^T X_i; 2M\lambda_2,\infty).
	\end{align}
	Thus, $\hat{\gamma}$ can be obtained in closed-form using \eqref{eq:case1dualsolexplicit}. Since $\hat{\beta}$ is the output of Algorithm 2, we know that the corresponding set $\mathcal{V}$ in Algorithm 2 is empty. 
	Thus, by Proposition \ref{prop:V_explicit}, we have $|\hat{r}^T X_i| \leq 2\sqrt{\lambda_0 \lambda_2}$ for any $i \in \text{Supp}(\hat{\beta})^c$. Using $\hat{r} = - \hat{\alpha}$ and $\sqrt{{\lambda_0}/{\lambda_2}} \leq  M$ in the previous in inequality, we get $|\hat{\alpha}^T X_i| \leq 2M\lambda_2$. Using the latter bound in \eqref{eq:case1dualsolexplicit}, we get that $\hat{\gamma}_i = 0$ for any $i \in \text{Supp}(\hat{\beta})^c$. However, for $i \in \text{Supp}(\hat{\beta})$, $\hat{\gamma}_i$ can be potentially nonzero.
	
	\smallskip

	\noindent{\bf{Case of $\sqrt{{\lambda_0}/{\lambda_2}} >  M$:} }  
	We construct a dual feasible solution ($\hat{\rho}, \hat{\mu}$) for~\eqref{eq:objdualcase2} as follows:
	\begin{align} \label{eq:case2dualsol}
	\hat{\rho} = - \hat{r} ~~ \text{and}
	~~ \hat{\mu} \in \argmax_{\mu \in \mathbb{R}^p} & ~~ h_2(\hat{\rho}, \mu)~~ \text{   s.t.   } ~~ (\hat{\rho}, \mu) \text { is feasible for } \eqref{eq:objdualcase2}.
	\end{align}
	Similar to \eqref{eq:case2dualsol}, the choice $\hat{\rho} = - \hat{r}$ is motivated by the optimality conditions in Theorem \ref{theorem:duals}, whereas the choice $\hat{\mu}$ maximizes the dual objective under the condition that $\rho = - \hat{r}$ (while ensuring feasibility). It can be readily seen that $\hat{\mu}_i$ is given in closed form by:
	\begin{align} \label{eq:mu_hat}
	\hat{\mu}_i = 
	\Big [|\hat{\rho}^T X_i| - \lambda_0/M - \lambda_2 M \Big]_{+}.
	\end{align}
	Note that for any $i \in \text{Supp}(\hat{\beta})^c$, we have $|\hat{\rho}^T X_i| = |\hat{r}^T X_i| \leq \lambda_0/M + \lambda_2 M$ (by Proposition \ref{prop:V_explicit}), which implies that $\hat{\mu}_i = 0$. However, for $i \in \text{Supp}(\hat{\beta})$, $\hat{\mu}_i$ can be potentially nonzero. 
	
	
	
	\smallskip
	
	\noindent\textbf{Quality of the Dual Bounds:} In Theorem~\ref{theorem:dualitygap}, we quantify the tightness of the dual bounds obtained from the dual feasible solutions \eqref{eq:case1dualsol} and \eqref{eq:case2dualsol}.
	\begin{theorem} \label{theorem:dualitygap}
		Let $\alpha^{*}$, $\gamma^{*}$, $\rho^{*}$, and $\mu^{*}$ be the optimal dual variables defined in Theorem \ref{theorem:duals}. Let $\beta^{*}$ be an optimal solution to \eqref{eq:relaxation}, and $\hat{\beta}$ be an inexact solution obtained using Algorithm 2. Define the primal gap $\epsilon = \| X (\beta^{*} - \hat{\beta}) \|_2$. Let $k = \| \hat{\beta} \|_0$ denote the number of nonzeros in the inexact solution to~\eqref{eq:relaxation}. For a fixed $(M,\lambda_2)$, the following holds for the dual solution $(\hat{\alpha}, \hat{\gamma})$ defined in~\eqref{eq:case1dualsol}:
		\begin{align} \label{eq:dualbd1}
		& h_1(\hat{\alpha}, \hat{\gamma}) \geq  h_1(\alpha^{*}, \gamma^{*}) - k \mathcal{O}(\epsilon) - k \mathcal{O}(\epsilon^2),
		\end{align}
		and for the dual solution $(\hat{\rho}, \hat{\mu})$ defined in \eqref{eq:case2dualsol}, we have:
		\begin{align} \label{eq:dualbd2}
		& h_2(\hat{\rho}, \hat{\mu}) \geq  h_2(\rho^{*}, \mu^{*}) - k \mathcal{O}(\epsilon) - \mathcal{O}(\epsilon^2). 
		\end{align}
	\end{theorem}
	Interestingly, the bounds established in Theorem \ref{theorem:dualitygap} do not depend on $p$, but rather on the support size $k$. Specifically, the constants in $\mathcal{O}(\epsilon)$ and $\mathcal{O}(\epsilon^2)$ only involve $M$ and $\lambda_2$ (these constants are made explicit in the proof). In practice, we seek highly sparse  solutions, i.e., $k \ll p$---suggesting that the quality of the dual bounds deteriorates with $k$ and not $p$. The main driver behind these tight bounds is Algorithm 2, which performs optimality checks on the coordinates outside the support. If vanilla CD was used instead of Algorithm 2, then the term $k$ appearing in the bounds~\eqref{eq:dualbd1} and \eqref{eq:dualbd2} will be replaced by $p$, making the bounds loose\footnote{This holds by using the argument in the proof of Theorem \ref{theorem:dualitygap} for Algorithm 1 instead of Algorithm 2.}.
	
	\smallskip
	
	\noindent\textbf{Efficient Computation of the Dual Bounds: }
	A direct computation of the dual bound $h_1(\hat{\alpha}, \hat{\gamma})$ or $h_2(\hat{\rho}, \hat{\mu})$ costs $\mathcal{O}(np)$ operations. 
	Interestingly, we show that this cost can be reduced to $\mathcal{O}(n + n \| \hat{\beta} \|_0)$ (where we recall that $\hat{\beta}$ is a solution from Algorithm 2). First, we consider the case of  $\sqrt{{\lambda_0}/{\lambda_2}} \leq  M $, where the goal is to compute $h_1(\hat{\alpha}, \hat{\gamma})$. By Lemma \ref{lemma:v} (in the appendix), we have $v(\hat{\alpha}, \hat{\gamma}_i) = 0$ for every $i \in \text{Supp}(\hat{\beta})^c$. Thus, $h_1(\hat{\alpha}, \hat{\gamma})$ can be simplified to: 
	\begin{align} \label{eq:h1_simplified}
	h_1(\hat{\alpha}, \hat{\gamma}) = - \frac{1}{2} \|\hat{\alpha} \|_2^2 - \hat{\alpha}^T y - \sum_{i \in \text{Supp}(\hat{\beta})} v(\hat{\alpha}, \hat{\gamma}_i).
	\end{align}
	Now, we consider the case of $\sqrt{{\lambda_0}/{\lambda_2}} >  M $, where the goal is to evaluate $h_2(\hat{\rho}, \hat{\mu}) $. By construction, the solution \eqref{eq:case2dualsol} is dual feasible, which means that the constraints in \eqref{eq:objdualcase2} need not be checked when computing the bound. Moreover, $\hat{\mu}_i = 0$ for every $i \in \text{Supp}(\hat{\beta})^c$ (see the discussion after \eqref{eq:mu_hat}). Thus, the dual bound can be expressed as follows: 
	\begin{align} \label{eq:h2_simplified}
	h_2(\hat{\rho}, \hat{\mu}) = - \frac{1}{2} \| \hat{\rho} \|_2^2 - \hat{\rho}^T y - M \sum_{i \in \text{Supp}(\hat{\beta})} |\hat{\mu}_i|.
	\end{align}
	The expressions in \eqref{eq:h1_simplified} and \eqref{eq:h2_simplified} can be computed in $\mathcal{O}(n + n \| \hat{\beta} \|_0)$ operations.
	
	\subsection{Branching and Incumbents} \label{sec:branching}
	Many branching strategies for BnB have been explored in the literature, e.g., random branching, strong branching, and pseudo-cost branching \cite{belotti2013mixed,bonami2011more,applegate1998solution,achterberg2005branching}. Among these strategies, strong branching has proven to be very effective in minimizing the size of the search tree \cite{applegate1998solution,bonami2011more,belotti2013mixed}. Strong branching selects the variable which leads to the maximum increase in the lower bounds of the children nodes. To select such a variable, two temporary node subproblems should be solved for every non-integral variable in the current relaxation. This can become a computational bottleneck, as each temporary subproblem involves solving a nonlinear optimization problem similar to \eqref{eq:relaxation}. To address this challenge, we use a fast (approximate) version of strong branching, in which we restrict the optimization in these temporary subproblems to the active set of the current node (instead of optimizing over all $p$ variables). In practice, this often leads to very similar search trees compared to exact strong branching, since the active set of the parent is typically close to the support of the children.

	We obtain the initial upper bound using \texttt{L0Learn} \cite{fastsubset}, which uses CD and efficient local search algorithms  to obtain good quality feasible  solutions for Problem \eqref{eq:conicbigM}. Moreover, at every node in the BnB tree, we attempt to improve the incumbent by making use of the support of a solution to the node relaxation. Specifically, we solve the following $\ell_2$ regularized least squares problem restricted to the relaxation's support $S$:
	$
	\min_{\beta_{{S}} \in \mathbb{R}^{|{S}|}} \frac{1}{2} \| y - X_{{S}} \beta_{{S}} \|^{2}_2 + \lambda_2 \|\beta_{{S}}\|^{2}_2.
	$
	Since ${S}$ is typically small, this problem can be solved efficiently by inverting a small $|S|\times|S|$ matrix. If $S$ is similar to the support of the parent node, then a solution for the current node can be computed via a low-rank update~\cite{hammarling2008updating}.

	\section{Experiments} \label{sec:experiments}
	We perform a series of high-dimensional experiments to study the run time of our BnB, understand its  sensitivity to parameter and algorithm choices, and compare to state-of-the-art approaches. While our dataset and parameter choices are motivated by statistical considerations, our focus here is not to study the statistical properties of $\ell_0$ estimators. We refer the reader to \cite{fastsubset} for empirical studies on statistical properties. 
	
	\subsection{Experimental Setup} \label{section:setup}
	
	\textbf{Synthetic Data Generation: } We generate a multivariate Gaussian data matrix with samples drawn from $\text{MVN}(0, \Sigma_{p \times p})$, a sparse coefficient vector $\beta^{\dagger} \in \mathbb{R}^p$ with $k^{\dagger}$ equi-spaced nonzero entries all set to $1$, and a noise vector $\epsilon_{i} \stackrel{\text{iid}}{\sim}  N(0, \sigma^2)$. We denote the support of the true regression coefficients $\beta^{\dagger}$ by $S^{\dagger}$. The response is then obtained from the linear model $y = X \beta^{\dagger} + \epsilon$. We define the \textsl{signal-to-noise ratio (SNR)} as follows: SNR$=\text{Var}(X \beta^{\dagger})/\sigma^2$. Unless otherwise specified, we set $\sigma^2$ to achieve SNR$=5$---this is a relatively difficult setting which still allows for full support recovery, under suitable choices of $n$, $p$, and $\Sigma$ (see \cite{fastsubset,mazumder2020discussion} for a discussion on appropriate levels of SNR). {We perform mean-centering followed by normalization (to have unit $\ell_2$ norm) on $y$ and each column of $X$. To help in exposition, we denote the resulting processed dataset by $(\tilde{y}, \tilde{X})$ and the (scaled) regression coefficients by $\tilde{\beta}^\dagger$.}
	
	\smallskip

	\noindent \textbf{Warm Starts, $\lambda_0$, $\lambda_2$, and $M$: } {The parameters $\lambda_2$ and $M$ can affect the run time significantly, so we study the sensitivity to these choices in our experiments. We consider choices that are relevant from  a statistical perspective, as we discuss next. For a fixed $\lambda_2$, let $\beta(\lambda_2)$ be the solution of ridge regression restricted to the support of the true solution $\tilde{\beta}^{\dagger}$:
		\begin{align}
		\beta(\lambda_2)  \in \argmin_{\beta \in \mathbb{R}^p} \frac{1}{2} \| \tilde{y} - \tilde{X} \beta \|^{2}_2  + \lambda_2 \|\beta\|^{2}_2 ~~~~ \text{s.t.} ~~~~ \beta_{(S^{\dagger})^c} = 0,
		\end{align}
		We define $\lambda_2^{*}$ as the $\lambda_2$ which minimizes the $\ell_2$ estimation error of $\beta(\lambda_2)$, i.e., 
		\begin{align} \label{eq:estimation_error}
		\lambda_2^{*} \in \argmin_{\lambda_2 \geq 0} \| \tilde{\beta}^{\dagger} - \beta(\lambda_2) \|_2.
		\end{align}
		We estimate $\lambda_2^{*}$ using grid search, with $50$ points equi-spaced on a logarithmic scale in the range $[10^{-4}, 10^{4}]$.
		In the experiments, we report our $\lambda_2$ choices as a fraction or multiple of $\lambda_2^{*}$ (e.g., $\lambda_2 = 0.1 \lambda_2^{*}$). Moreover, for each $\lambda_2$, we define  $M^{*}(\lambda_2) = \| \beta(\lambda_2) \|_{\infty}$ and report our choices of the Big-M in terms of $M^{*}(\lambda_2)$---we also use the notation $M^{*}$ to refer to $M^{*}(\lambda_2)$ when $\lambda_2$ is clear from context. Note that if Problem \eqref{eq:main} has a unique solution, and this solution has support $S^{\dagger}$, then $M^{*}(\lambda_2)$ is the smallest valid choice of $M$ in formulation \eqref{eq:conicbigM}. Unless otherwise specified, for each $\lambda_2$ considered, we fix $\lambda_0$ to a value $\lambda_0^{*}$ (which is a function of $\lambda_2$) that leads to the true support size  $k^{\dagger}$. We obtain $\lambda_0^{*}$ using \texttt{L0Learn}\footnote{\texttt{L0Learn} computes a data-dependent grid of $\lambda_0$ values. When CD is used to optimize \eqref{eq:main}, each $\lambda_0$ in the grid leads to a different support size. See Section 4.1 of  \cite{fastsubset} for more details.} \cite{fastsubset}. Note that $\lambda_0^{*}$ might not exist in general, but in all the experiments we considered, \texttt{L0Learn} was able to such a value. Moreover, Section \ref{sec:vary_grid}, presents an experiment where $\lambda_0$ is varied over a range that is independent of \texttt{L0Learn}. \\
		We obtain warm starts from \texttt{L0Learn}\footnote{{We use the default CD-based algorithm in \texttt{L0Learn}, which does not depend on any Big-M parameter. We remind the reader that \texttt{L0Learn} is a local optimization method that does not provide certificates of global optimality.}}  and use them for all the MIP solvers considered.}
	
	\smallskip
	
	\noindent \textbf{Solvers and Settings: } Our solver, L0BnB\footnote{\url{https://github.com/alisaab/l0bnb}}, is written in Python with critical code sections optimized using Numba \cite{lam2015numba}. We compare L0BnB with Gurobi (GRB), MOSEK (MSK), and BARON (B) on formulation \eqref{eq:conicbigM}. We also compare with \cite{bertsimas2017sparse} who solve the cardinality-constrained variant of \eqref{eq:main} with the number of nonzeros set to $k^{\dagger}$. 
	{In all solvers (including L0BnB), we use the following settings:
		\begin{itemize}
			\item Relative Optimality Gap: Given an upper bound UB and a lower bound LB, the relative optimality gap is defined as $(\text{UB} - \text{LB})/\text{UB}$. We set this to $1\%$.
			\item Integer Feasibility Tolerance ($\epsilon_{\text{if}}$): This tolerance is used to determine whether a variable obtained from a node relaxation will be declared as integral (for the purpose of branching). Specifically, in our context, if a variable $z_i$ is within $\epsilon_{\text{if}}$ from 0 or 1, then it is declared as integral. We set $\epsilon_{\text{if}} = 10^{-4}$.
			\item Primal-Dual Optimality Gap ($\epsilon_{\text{pd}}$): This is the relative gap between the primal and dual bounds at a given node, which is used as termination criterion for the subproblem solver. We set $\epsilon_{\text{pd}} = 10^{-5}$. In L0BnB, this gap is satisfied when an integral solution to a node's relaxation is encountered.
		\end{itemize}
	}
	
	\noindent {\textsl{Gradient Screening in L0BnB: } As discussed in Section \ref{sec:gradient_screening}, our gradient screening implementation leads to noticeable speed-ups for large $p$ ($\geq 10^4$ in our experience). For smaller values of $p$, we do not observe run time improvements.
		In the experiment of Section \ref{sec:timing_comparison} (Table \ref{table:gradinet_screening}), where we vary $p$, we report the running time with and without gradient screening. In the other experiments, we do not use gradient screening.}
	
	\smallskip
	
	\noindent {\textbf{Reproducibility and Additional Details}: In the spirit of reproducibility, the function used to generate and process the synthetic datasets is publicly available on L0BnB's Github page. In all experiments, we report the values of $M$, $\lambda_0$, and $\lambda_2$ (either in the main text or in Appendix \ref{appendix:additional_results}). Moreover, for all approaches, we report the number of BnB nodes explored in Appendix \ref{appendix:additional_results}.}
	

	\subsection{Comparison with State-of-the-art solvers}
	
	\subsubsection{Varying Number of Features}
	\label{sec:timing_comparison}
	In this section, we study the scalability of the different solvers in terms of the number of features $p$. We generate synthetic datasets\footnote{To ensure that the same estimates of  $\lambda_2^{*}$ and $M^{*}$ are used for the different choices of $p$, we generate a single data matrix with $10^{6}$ features. For each $p$, we take a submatrix that consists of all the rows and a subset of $p$ columns (which includes the true support).} with $n=10^{3}$, $p \in \{ 10^3,10^4,10^5,10^6\}$, and $k^{\dagger} =10$. We consider a constant correlation setting, where $\Sigma_{ij} = 0.1 \ \forall i \neq j$ and $1$ otherwise. The parameters  $\lambda_2$ and $M$ can have a significant effect on the run time. Thus, we report the timings for different choices of these parameters. In particular, in Table  \ref{table:fixed_M} (top panel), we report the timings for $\lambda_2 \in \{ 0.1\lambda_2^{*}, \lambda_2^{*}, 10\lambda_2^{*} \}$, where for each $\lambda_2$, we fix $M = 1.5 M^{*}(\lambda_2)$ (where $\lambda_2^{*}$ and $M^{*}$ are defined in Section \ref{section:setup}). In Table \ref{table:fixed_M} (bottom panel), we fix $\lambda_2 = \lambda_2^{*}$ and report the timings for $M \in \{ M^{*}(\lambda_2), 2M^{*}(\lambda_2), 4M^{*}(\lambda_2), \infty \}$. Note that the results for L0BnB in Table \ref{table:fixed_M} are without gradient screening; in Table \ref{table:gradinet_screening} in Appendix \ref{appendix:experiment_vary_p_m}, we report the timings with gradient screening enabled. {Moreover, in Appendix \ref{appendix:experiment_vary_p_m}, we report the values of $\lambda_0^{*}$ and $\lambda_2^{*}$, along with the number of BnB nodes explored.}
	\begin{table}[h]
		\caption{{\small (Sensitivity to $\lambda_{2}$ and $M$) Running time in seconds for solving \eqref{eq:main} (via formulation~\eqref{eq:conicbigM}) by our proposal (L0BnB), Gurobi (GRB), MOSEK (MSK), and BARON (B). The method~\cite{bertsimas2017sparse} solves the cardinality-constrained variant of~\eqref{eq:main}. For methods that do not terminate in 4 hours: the optimality gap is shown in parenthesis, and a dash (-) is used in the special case of a 100\% gap. {[Top panel] For each $\lambda_2$, we use $M = 1.5 M^{*}(\lambda_2)$. [Bottom panel] We use $\lambda_2 = \lambda_2^{*}$, and four different values of $M$. 
					The Big-M values are based on: $M^{*}(\lambda_2^{*}) = 0.232$, $M^{*}(0.1 \lambda_2^{*}) = 0.164$, and $M^{*}(10 \lambda_2^{*}) = 0.243$.
					The true solution satisfies: $\|\tilde{\beta}^{\dagger}\|_{\infty} = 0.208$. In the bottom panel, we found that PR($M^{*}$) and  PR($\infty$) have the same optimal solution.} }}
		\label{table:fixed_M}
		\centering
		\renewcommand{\tabcolsep}{1.4pt}
		\renewcommand{\arraystretch}{1.1}
		\begin{tabular}{l|ccccc|ccccc|ccccc|}
			\cline{2-16}
			& \multicolumn{5}{c|}{$\lambda_2 = \lambda_2^{*}$}  & \multicolumn{5}{c|}{$\lambda_2 = 0.1 \lambda_2^{*}$} & \multicolumn{5}{c|}{$\lambda_2 = 10 \lambda_2^{*}$}    \\ \hline
			\multicolumn{1}{|l|}{p}      & L0BnB         & GRB     & MSK  & B    & \cite{bertsimas2017sparse}      & L0BnB              & GRB       & MSK    & B  & \cite{bertsimas2017sparse}  & L0BnB         & GRB     & MSK & B    & \cite{bertsimas2017sparse}            \\ \hline
			\multicolumn{1}{|l|}{$10^3$} & \textbf{0.7}  & 70      & 92   & (4\%) & (34\%) & \textbf{2}         & 57        & 154    & -    & -   & \textbf{0.01}          & 148     & 28  & (5\%) & {0.08} \\
			\multicolumn{1}{|l|}{$10^4$} & \textbf{3}    & (15\%) & 1697 & -      & (78\%) & \textbf{12}        & (12\%)   & 3872   & -    & -   & \textbf{0.06} & (11\%) & 314 & -      & 5             \\
			\multicolumn{1}{|l|}{$10^5$} & \textbf{34}   & -       & -    & -      & (86\%) & \textbf{545}      & -         & -      & -    & -   & \textbf{0.5}  & -       & -   & -      & 8             \\
			\multicolumn{1}{|l|}{$10^6$} & \textbf{1112} & -       & -    & -      & -       & \textbf{(23\%)}   & -         & -      & -    & -   & \textbf{8}    & -       & -   & -      & 46            \\ \hline
		\end{tabular}
		\newline
		\vspace*{0.2cm}
		\newline
		\renewcommand{\tabcolsep}{1.2pt}
		\renewcommand{\arraystretch}{1.1}
		\begin{tabular}{l|cccc|cccc|cccc|cccc|}
			\cline{2-17}
			& \multicolumn{4}{c|}{$M = M^{*}$}    & \multicolumn{4}{c|}{$M =2 M^{*}$}     & \multicolumn{4}{c|}{$M = 4 M^{*}$}    & \multicolumn{4}{c|}{$M = \infty$}  \\ \hline
			\multicolumn{1}{|l|}{p}      & L0BnB        & GRB     & MSK   & B      & L0BnB            & GRB       & MSK    & B & L0BnB            & GRB       & MSK    & B & L0BnB            & GRB    & MSK    & B \\ \hline
			\multicolumn{1}{|l|}{$10^3$} & \textbf{0.2} & 27    & 57  & (2\%) & \textbf{0.8}  & 112    & 128  & - & \textbf{0.8}   & 1219   & 137  & - & \textbf{0.8}    & 3974 & 91   & - \\
			\multicolumn{1}{|l|}{$10^4$} & \textbf{0.9} & 10571 & 665 & -     & \textbf{5}    & (24\%) & 3260 & - & \textbf{5}     & (50\%) & 3289 & - & \textbf{6}      & -    & 9025 & - \\
			\multicolumn{1}{|l|}{$10^5$} & \textbf{8}   & -     & -   & -     & \textbf{70}   & -      & -    & - & \textbf{81}    & -      & -    & - & \textbf{81}     & -    & -    & - \\
			\multicolumn{1}{|l|}{$10^6$} & \textbf{121} & -     & -   & -     & \textbf{9265} & -      & -    & - & \textbf{10986} & -      & -    & - & \textbf{11010}  & -    & -    & - \\ \hline
		\end{tabular}
	\end{table}
	\raggedbottom

	The results in the top panel of Table \ref{table:fixed_M} indicate significant speed-ups, reaching over $200,000$x compared to Gurobi, $28,000$x compared to MOSEK, and $20,000$x compared to \cite{bertsimas2017sparse}. At $\lambda_2 = \lambda_2^{*}$, L0BnB is the only solver that can handle $p \geq 10^5$, and can, in fact, handle $p=10^6$ in less than 20 minutes. Recall that $\lambda_2^{*}$ minimizes the $\ell_2$ estimation error and leads to an estimator that is the closest to the ground truth. 
	When the ridge parameter is weak i.e., for $\lambda_2 = 0.1 \lambda_2^{*}$, L0BnB is again the only solver that can handle $p \geq 10^5$. When the ridge parameter is large, i.e., for $\lambda_2 = 10 \lambda_2^{*}$, the optimization problem seems to become easier: L0BnB can be more than  $100$x faster compared to $\lambda_2^{*}$, and  \cite{bertsimas2017sparse} can handle up to $p \approx 10^6$.  The speed-ups for $\lambda_2 = 10 \lambda_2^{*}$ can be attributed to the fact that a larger $\lambda_2$ adds a large amount of regularization to the objective (via the perspective term)---improving the performance of the relaxation solvers\footnote{Gurobi is the only exception to this observation. We investigated this: Gurobi generates additional cuts only for the case of $10\lambda_2^{*}$, which seem to slow down the relaxation solver.}.  It is also worth emphasizing that L0BnB is prototyped in Python, as opposed to the highly efficient BnB routines available in commercial solvers such as Gurobi and MOSEK. 
	
	Ideally, we desire a solver that can solve Problem~\eqref{eq:main} over a range of $\lambda_2$ values, which includes values in the neighborhood of $\lambda^{*}_2$. However, the results in Table \ref{table:fixed_M} suggest that the state-of-the-art methods (except L0BnB) seem to only work for quite large values of $\lambda_2$ (which, in this case, do not correspond to solutions that are interesting from a statistical viewpoint).
	On the other hand, L0BnB seems to be the only method that can scale to $p \sim 10^6$ while being relatively robust to the choice of $\lambda_2$. 
	
	In the bottom panel of Table \ref{table:fixed_M}, the results also indicate that L0BnB significantly outperforms Gurobi, MOSEK, and BARON for different choices of $M$. For all the solvers, the run time increases with $M$, and the longest run times are for $M=\infty$. This empirically validates our result in Proposition \ref{prop:v1v2}, where we show that for a sufficiently small (but valid) value of $M$, PR($M$) can be better than PR$(\infty)$, in terms of the relaxation quality. However, even with $M=\infty$, L0BnB can solve $p=10^6$ in around 3 hours, whereas all other solvers have a 100\% gap after 4 hours. {We also note that Table~\ref{table:fixed_M} considers both the two settings in Theorem~\ref{theorem:relaxation}: $\sqrt{\lambda_0/\lambda_2} > M$ and $\sqrt{\lambda_0/\lambda_2} \leq M$. Specifically, we have $\sqrt{\lambda_0^{*}/\lambda_2^{*}} > M$ for $M \in \{ M^{*}(\lambda_2^{*}), 1.5 M^{*}(\lambda_2^{*}), 2 M^{*}(\lambda_2^{*}) \}$, and $\sqrt{\lambda_0^{*}/\lambda_2^{*}} \leq M$ for $M \in \{ 4 M^{*}(\lambda_2^{*}), \infty  \}$. L0BnB achieves notable improvements over other solvers for both of these settings.}

	{\subsubsection{Varying Signal-to-Noise Ratio (SNR)} \label{sec:vary_snr}
		Here we investigate the effect of SNR on the running time of the different solvers. To this end, we generate synthetic datasets with $n=10^{3}$, $p=10^3$, $k^{\dagger} =10$, under a constant correlation setting, where $\Sigma_{ij} = 0.1 \ \forall i \neq j$ and $1$ otherwise. We vary SNR in $\{ 0.5, 1, 2, 3, 4, 5 \}$. We set $\lambda_0 = \lambda_0^{*}$,   $\lambda_2 = \lambda_2^{*}$ (where $\lambda_0^{*}$ and $\lambda_2^{*}$ depend on SNR), and $M = 1.5 M^{*} (\lambda_2)$. We report the running time versus SNR in Table \ref{table:vary_snr}. The corresponding values of $\lambda_0^{*}$ and $\lambda_2^{*}$, and the number of BnB nodes are reported in Appendix \ref{appendix:vary_snr}. The results indicate that L0BnB is much faster, and less sensitive to changes in SNR, compared to the other solvers. For example, L0BnB can handle SNR=0.5 in less than 4 minutes, whereas the fastest competing solver (MOSEK) takes around 46 minutes.}
	
	
	\begin{table}[h]
		\caption{{\small Running time (seconds) for solving~\eqref{eq:conicbigM} on a synthetic dataset with $n=p=10^3$, at different SNR levels. The parameter $\lambda_2$ is set to the optimal choice $\lambda_2^{*}$, which depends on the current SNR level. For methods that do not terminate in 2 hours: the optimality gap is shown in parenthesis, and a dash (-) is used in the special case of a 100\% gap.}}
		\label{table:vary_snr}
		\centering
		\renewcommand{\tabcolsep}{1.8pt}
		\renewcommand{\arraystretch}{1.3}
		{\begin{tabular}{|c|c|cccc|}
				\hline
				SNR & M    & L0BnB        & GRB         & MSK & B \\ \hline
				0.5 & 0.269   & \textbf{231} & {(2\%)} & 2757  & -     \\
				1   & 0.307  & \textbf{1.7} & 1213           & 1046  & -     \\
				2   & 0.333  & \textbf{1.4} & 119            & 288   & -     \\
				3   & 0.346    & \textbf{1.3} & 102            & 173   & -     \\
				4   & 0.347   & \textbf{1.0} & 77             & 113   & -     \\
				5   & 0.348    & \textbf{0.8} & 77             & 92    & -     \\ \hline
		\end{tabular}}
	\end{table}

	\subsubsection{Varying $\lambda_0$ and $\lambda_2$} \label{sec:vary_grid}
	{In the experiment of Section \ref{sec:timing_comparison}, we studied the running time for  $\lambda_0 = \lambda_0^{*}$ so that the model recovers the true support size. In this experiment, we study the running time over a grid of $\lambda_0$ and $\lambda_2$ values, which includes various support sizes. We consider  synthetic instances with $p \in \{ 10^3, 10^4 \}$, $n=10^3$, $k^{\dagger} = 10$, under a constant correlation setting, where $\Sigma_{ij} = 0.1 \ \forall i \neq j$ and $1$ otherwise.} {We vary $\lambda_2 \in \{ 0.1, 0.5, 1, 2, 10\} \cdot \lambda_2^{*}$ and $\lambda_0 \in \{ 0.5 , 0.1 , 0.01 \} \cdot \lambda_0^{\text{m}}$, where $\lambda_0^{\text{m}}$ is a value of $\lambda_0$ which sets all coefficients to zero. Following \cite{fastsubset}\footnote{\cite{fastsubset} shows that $\lambda_0^{\text{m}}$ is the smallest choice of $\lambda_0$ for which $\beta = 0$ is a coordinate-wise minimizer for Problem \eqref{eq:main}. In this experiment, we verified numerically that $\beta = 0$ is a global minimizer at $\lambda_0^{\text{m}}$.}, we use $\lambda_0^{\text{m}} = {(2 + 4 \lambda_2)}^{-1} \| X^T y \|_{\infty}^2$. 
		Next, we describe, how given a pair $(\lambda_0, \lambda_2)$ in the grid, we estimate the corresponding  Big-M value $M(\lambda_0, \lambda_2)$. Let $\beta(\lambda_0, \lambda_2)$ be the solution of Problem \eqref{eq:main} restricted to the true support:
		$$
		\beta(\lambda_0, \lambda_2) \in \argmin_{\beta} \frac{1}{2} \| y - X \beta \|^{2}_2  + \lambda_0 \| \beta \|_0+ \lambda_2 \|\beta\|^{2}_2 ~~ \text{s.t.} ~~ \beta_{(S^{\dagger})^c} = 0.
		$$
		We compute $\beta(\lambda_0, \lambda_2)$ exactly using formulation \eqref{eq:conicbigM} with $M = \infty$. We then compute an estimate of the  Big-M value: $\tilde{M}(\lambda_0, \lambda_2) := \|\beta(\lambda_0, \lambda_2)\|_{\infty}$. For every $(\lambda_0, \lambda_2)$ in the grid, we 
		solve Problem \eqref{eq:conicbigM} for $M(\lambda_0, \lambda_2) = 1.5 \tilde{M}(\lambda_0, \lambda_2)$. In Table \ref{table:vary_l0_l2}, we report the running time for $p=10^3$ (top panel) and $p=10^4$ (bottom panel). The values of $\lambda_2^{*}$ and $\lambda_0$, along with the number of BnB nodes explored, are reported in Appendix \ref{appendix:vary_l0_l2}.}
	
	{In Table \ref{table:vary_l0_l2}, for all values of $\lambda_2$ considered and $\lambda_0 \in \{ 0.5 , 0.1 \} \cdot  \lambda_0^{\text{m}}$, L0BnB is faster than the competing methods, with speed-ups exceeding 10,000x compared to both Gurobi and MOSEK. However, for $\lambda_0 = 0.01 \lambda_0^{\text{m}}$, none of the solvers are able to solve the problem in 1 hour, and the gaps seem to be comparable. We also note that for $p=10^4$, L0BnB is able to solve (to optimality) 13 out of the 20 instances in the 1 hour limit. In contrast, Gurobi could not solve any of the instances and MOSEK could only solve 6.}

	\begin{table}[htbp]
		\caption{{\small Running time for solving \eqref{eq:conicbigM} over a grid of $\lambda_0$ and $\lambda_2$ values. ``NNZ'' is the number of nonzeros in the solution to \eqref{eq:conicbigM} (or the best incumbent if all solvers cannot terminate in 1 hour). For methods that do not terminate in 1 hour: the optimality gap is shown in parenthesis, and a dash (-) is used in the special case of a 100\% gap. The true solution satisfies: $\|\tilde{\beta}^{\dagger}\|_{\infty} = 0.208$.}}
		\label{table:vary_l0_l2}
		\centering
		\qquad \qquad 
		{\begin{tabular}{|cccc|cccc|}
				\hline
				\multicolumn{8}{|c|}{${p=10^3}$} \\
				\hline
				$\lambda_2$                          & $\lambda_0$                 & NNZ & M     & L0BnB            & GRB              & MSK              & B \\ \hline
				\multirow{3}{*}{$10 \lambda_2^{*}$}  
				& $0.5 \lambda_0^{\text{m}}$  & 5   & 0.293 & \textbf{1}       & (3\%)           & 158              & - \\
				& $0.1 \lambda_0^{\text{m}}$  & 10  & 0.245 & \textbf{0.01}    & (5\%)           & 5                & - \\
				& $0.01 \lambda_0^{\text{m}}$ & 36  & 0.245 & \textbf{(12\%)} & (17\%)          & (16\%)          & - \\ \hline
				\multirow{3}{*}{$2 \lambda_2^{*}$}   
				& $0.5 \lambda_0^{\text{m}}$  & 4   & 0.491 & \textbf{3}       & 651              & 2737             & - \\
				& $0.1 \lambda_0^{\text{m}}$  & 10  & 0.332 & \textbf{0.1}     & 68               & 106              & - \\
				& $0.01 \lambda_0^{\text{m}}$ & 15  & 0.332 & \textbf{(3\%)}  & (5\%)           & (4\%)           & - \\ \hline
				\multirow{3}{*}{$ \lambda_2^{*}$}    
				& $0.5 \lambda_0^{\text{m}}$  & 3   & 0.524 & \textbf{20}      & 2648             & 1278             & - \\
				& $0.1 \lambda_0^{\text{m}}$  & 10  & 0.348 & \textbf{0.3}     & 109              & 65               & - \\
				& $0.01 \lambda_0^{\text{m}}$ & 10  & 0.348 &       {(10\%)} & (11\%)          & \textbf{(6\%)}  & - \\ \hline
				\multirow{3}{*}{$0.5 \lambda_2^{*}$} 
				& $0.5 \lambda_0^{\text{m}}$  & 3   & 0.542 & \textbf{59}      & 965              & (6\%)           & - \\
				& $0.1 \lambda_0^{\text{m}}$  & 10  & 0.356 & \textbf{1}       & 63               & 172              & - \\
				& $0.01 \lambda_0^{\text{m}}$ & 10  & 0.356 & {(19\%)} & (16\%)            & \textbf{(12\%)} & - \\ \hline
				\multirow{3}{*}{$0.1 \lambda_2^{*}$} 
				& $0.5 \lambda_0^{\text{m}}$  & 3   & 0.557 & \textbf{128}     & 697              & (8\%)           & - \\
				& $0.1 \lambda_0^{\text{m}}$  & 10  & 0.364 & \textbf{1}       & 58               & 303              & - \\
				& $0.01 \lambda_0^{\text{m}}$ & 10  & 0.364 &  {(47\%)} & \textbf{(32\%)} & (55\%)          & - \\ \hline
		\end{tabular}}
		\newline
		\vspace*{0.2cm}
		\newline
		{\begin{tabular}{|cccc|cccc|}
				\hline
				\multicolumn{8}{|c|}{${p=10^4}$} \\
				\hline
				$\lambda_2$                          & $\lambda_0$                 & NNZ & M     & L0BnB            & GRB        & MSK              & B \\ \hline
				\multirow{3}{*}{$10 \lambda_2^{*}$} 
				& $0.5 \lambda_0^{\text{m}}$  & 5   & 0.293 & \textbf{2}       & -          & 1617             & - \\
				& $0.1 \lambda_0^{\text{m}}$  & 10  & 0.245 & \textbf{0.1}       & -          & 62               & - \\
				& $0.01 \lambda_0^{\text{m}}$ & 45  & 0.245 & \textbf{(3\%)}  & -          & (7\%)           & - \\ \hline
				\multirow{3}{*}{$2 \lambda_2^{*}$}  
				& $0.5 \lambda_0^{\text{m}}$  & 8   & 0.491 & \textbf{263}     & -          & (13\%)          & - \\
				& $0.1 \lambda_0^{\text{m}}$  & 10  & 0.332 & \textbf{0.7}     & -          & 2459             & - \\
				& $0.01 \lambda_0^{\text{m}}$ & 17  & 0.332 & \textbf{(15\%)} & -          & \textbf{(15\%)}          & - \\ \hline
				\multirow{3}{*}{$ \lambda_2^{*}$}   
				& $0.5 \lambda_0^{\text{m}}$  & 3   & 0.524 & \textbf{1920}    & -          & (18\%)          & - \\
				& $0.1 \lambda_0^{\text{m}}$  & 10  & 0.348 & \textbf{2}       & -          & 3338             & - \\
				& $0.01 \lambda_0^{\text{m}}$ & 14  & 0.348 & {(29\%)} & -          & \textbf{(24\%)} & - \\ \hline
				\multirow{3}{*}{$0.5 \lambda_2^{*}$}
				& $0.5 \lambda_0^{\text{m}}$  & 3   & 0.542 & \textbf{(3\%)}  & (29\%)    & (26\%)          & - \\
				& $0.1 \lambda_0^{\text{m}}$  & 10  & 0.356 & \textbf{5}       & -          & 3483             & - \\
				& $0.01 \lambda_0^{\text{m}}$ & 14  & 0.356 & {(42\%)} & -          & \textbf{(36\%)} & - \\ \hline
				\multirow{3}{*}{$0.1 \lambda_2^{*}$}
				& $0.5 \lambda_0^{\text{m}}$  & 3   & 0.557 & \textbf{(8\%)}  & -          & (19\%)            & - \\
				& $0.1 \lambda_0^{\text{m}}$  & 10  & 0.364 & \textbf{24}      & -          & (46\%)          & - \\
				& $0.01 \lambda_0^{\text{m}}$ & 13  & 0.364 & \textbf{(73\%)} & - & (77\%)          & - \\ \hline
		\end{tabular}}
	\end{table}

	\subsection{Sensitivity to Data Parameters} \label{sec:sensitivity}
	We study how L0BnB's running time is affected by the following data specific parameters: number of samples ($n$), feature correlations ($\Sigma$), and number of nonzero coefficients ($k^{\dagger}$). In the experiments below, we fix $\lambda_2 = \lambda_2^{*}$ and $M = 1.5 M^{*}(\lambda_2^{*})$. {For all the experiments in this section, the values of $\lambda_0$, $\lambda_2$, and $M$ used are reported in Appendix \ref{appendix:L0BnB_exps}.}
	
	\smallskip

	\noindent \textbf{Number of Samples: } We fix $k^{\dagger} = 10$, $p = 10^4$, and consider a constant correlation setting $\Sigma_{ij} = 0.1$ for $i \neq j$ and $1$ otherwise\footnote{We choose $p=10^4$ to ensure that the matrix fits into memory when $n$ is large.}. The timings for $n \in \{10^2, 10^3, 10^4, 10^5\}$ are in Table \ref{table:sensitivity}. The results indicate that the problem can be solved in reasonable times (order of seconds to minutes) for $n \geq 10^3$. The problems can be solved the fastest when $n$ is close to $p$, and the extreme cases $n=10^2$ and $n=10^5$ are the slowest. We contend that for large $n$, the CD updates (which cost $\mathcal{O}(n)$ each) become a bottleneck. For $n=10^2$, the CD updates are cheaper, however, the size of the search tree is significantly larger---suggesting that a small value of $n$ can lead to loose relaxations and require more branching. Note also that the underlying statistical problem is the most difficult for $n=10^2$ compared to other larger values of $n$.
	
	\smallskip

	\noindent \textbf{Feature Correlations: } We generate synthetic datasets with $p \in \{10^4, 10^5 \}$, $n=10^3$, $k^{\dagger} = 10$, and $\Sigma = \text{Diag}(\Sigma^{(1)}, \Sigma^{(2)}, \dots, \Sigma^{(k^{\dagger})} )$ is block-diagonal. Such block structures are commonly used in the sparse learning literature~\cite{zhao2006model}. Each $\Sigma^{(l)}$ is a correlation matrix with the same dimension. Given a correlation parameter $\rho \in (0,1)$, for each $l \in [k^{\dagger}]$, we assume that $\Sigma^{(l)}_{ij} = \rho^{|i-j|}$ for any $i,j \in [p/k^{\dagger}]$, i.e., the correlation is exponentially decaying in each block. We report the timings for different values of the parameter $\rho$ in Table \ref{table:sensitivity}. The results indicate that higher correlations lead to an increase in the run time, but  even the high correlation instances can be solved in reasonable time. For example, when $p=10^4$, $\rho=0.9$ can be solved in less than 10 minutes. When $p=10^5$, $\rho=0.8$ can be solved in around 13 minutes, and $\rho=0.9$ can be solved to a 13\% gap in 2 hours.
	
	\smallskip
	
	\noindent \textbf{Sparsity of the Regression Coefficients: } We consider datasets with $n=10^3$, $p=10^4$, and the same correlation setting as the experiment for the number of samples. The results in Table \ref{table:sensitivity} show that problems with 15 nonzeros can be handled in $166$ seconds, and for $20$ and $25$ nonzeros, decent gaps ($\leq 10\%$) can be obtained in 2 hours. We also note that for larger $\lambda_2$ or tighter $M$ choices, larger values of $k^{\dagger}$ can be handled.
	\begin{table}[h]
		\renewcommand{\tabcolsep}{1.2pt}
		\renewcommand{\arraystretch}{1.5}
		\centering
		\caption{{\small{Run time in seconds (denoted by t) for L0BnB. If L0BnB  does not solve the problem in a 2 hour time limit, the optimality gap is shown in parenthesis.}}}
		\label{table:sensitivity}
		\begin{tabular}{|c|c|c|c|c|}
			\multicolumn{5}{c}{Varying $n$} \\
			\hline
			$\boldsymbol n$ & $10^2$ & $10^3$ & $10^4$ & $10^5$ \\ \hline
			\bf{t} & (42\%) & 11 & 30 & 1701 \\ \hline
		\end{tabular}
		\quad
		\begin{tabular}{|c|c|c|c|c|c|c|}
			\multicolumn{7}{c}{Varying correlation coefficient  $\rho$} \\
			\hline
			$\boldsymbol \rho$ & 0.1	& 0.3	& 0.5	& 0.7	& 0.8	& 0.9 \\ \hline
			\bf{t}($p=10^4$) & 4 & 4  & 5 & 16 & 55 & 530 \\ \hline
			\bf{t}($p=10^5$) & 35 & 39  & 60 & 231 &	808 & (13\%) \\ \hline
		\end{tabular}
		\quad
		\begin{tabular}{|c|c|c|c|c|}
			\multicolumn{5}{c}{Varying $k^{\dagger}$} \\
			\hline
			$\boldsymbol k^{\dagger}$ & 10	& 15	& 20	& 25 \\ \hline
			\bf{t}	& 4	& 166	& (10\%)	& (6\%) \\ \hline
		\end{tabular}
	\end{table}
	
	\subsection{Real Data and Ablation Studies (Algorithm Settings)}
	
	\noindent{\bf{High-dimensional Real Data}:} Here we investigate the run time of L0BnB on the Riboflavin dataset \cite{buhlmann2014high}---a genetics dataset used for predicting Vitamin B2 production levels. The original dataset has $p = 4088$ and $n = 71$. We augment the dataset with pairwise feature interactions to get $p = 8,362,004$. {We mean center and normalize the response and the columns of the data matrix}.  We then run 5-fold cross-validation in \texttt{L0Learn} (with default parameters) to find the optimal regularization parameters $\hat{\lambda}_0$ and $\hat{\lambda}_2$. We set $M$ to be $10$ times the $\ell_\infty$ norm of the solution obtained from cross-validation. In L0BnB, we solve the problem for $\lambda_2 \in \{ 0.1  \hat{\lambda}_2, \hat{\lambda}_2\}$ and vary $\lambda_0$ to generate solutions of different support sizes. The run time, for each $\lambda_2$, as a function of the support size is reported in Figure \ref{fig:real}. Interestingly, at $\hat{\lambda}_2$, all the support sizes obtained (up to 15 nonzeros) can be handled in less than a minute. Moreover, the increase in time is relatively slow as the number of nonzeros increases. When $\lambda_2$ becomes smaller (i.e., $\lambda_{2}=0.1  \hat{\lambda}_2$), the run times increase (though they are reasonable) as the problem becomes more difficult. When an optimal solution has six nonzeros, L0BnB for $\lambda_{2}=0.1\hat{\lambda}_2$ is approximately $20$x slower than $\lambda_{2}=\hat{\lambda}_2$.
	
	
	
	\smallskip
	
	\noindent{\bf{Ablation Study}:} We perform an ablation study to measure the effect of key choices in our BnB on the run time. Particularly, we consider the following changes: (i) replacing our relaxation solver with MOSEK, (ii) turning off warm starts in our solver, and (iii) turning off active sets in our solver. We measure the run time before and after these changes on synthetic data with $n=10^3$, $k^{\dagger} = 10$, $\Sigma_{ij} = 0.1$ for all $i \neq j$ and $1$ otherwise. We use the parameters $\lambda_0 = \lambda_0^{*}$,  $\lambda_2 = \lambda_2^{*}$, and $M = 1.5 M^{*}(\lambda_2^{*})$ (with the same values as those used in the experiment of Section \ref{sec:timing_comparison}). The run times for different choices of $p$ are reported in Table \ref{table:ablation}. 
	The results show that replacing our relaxation solver with MOSEK will slow down the BnB by more than $1200$x at $p=10^4$. The results for MOSEK are likely to be even slower for $p=10^6$ as it still had a $100\%$ gap after 2 hours. We note that MOSEK employs a state-of-the-art conic optimization solver (based on an interior point method). The significant speed-ups here can be attributed to our CD-based algorithm, which is designed to effectively exploit the sparsity structure in the problem. The results also indicate that warm starts and active sets are important for run time, e.g., removing warm starts and active sets at $p=10^5$ can slow down the algorithm by more than $16$x and $67$x, respectively.

	\begin{figure}[h]
		\begin{minipage}[c]{0.5\textwidth}
			\renewcommand{\tabcolsep}{4pt}
			\renewcommand{\arraystretch}{1.3}
			\centering
			\captionof{table}{{\small Time (seconds) after the following changes to our relaxation solver: (i) replacing it with MOSEK, (ii) removing warm starts, and (iii) removing active sets. Gap is shown in parenthesis if the method does not terminate in 2 hours.}}
			\label{table:ablation}
			\scalebox{1.05}{\begin{tabular}{|c| ccc|}
					\hline
					$p$ &  $10^4$ & $10^5$ & $10^6$ \\ \hline
					L0BnB &  3 & 34 & 1112 \\
					(i)   & 3802 & (13\%) & (100\%) \\ 
					(ii)    & 25 & 560 & (21\%) \\
					(iii) & 66 & 2291 & (23\%) \\\hline
			\end{tabular}}
		\end{minipage}
		\begin{minipage}[c]{0.5\textwidth}
			\centering
			\includegraphics[scale=0.5,trim=2mm 2mm 2mm 1mm, clip]{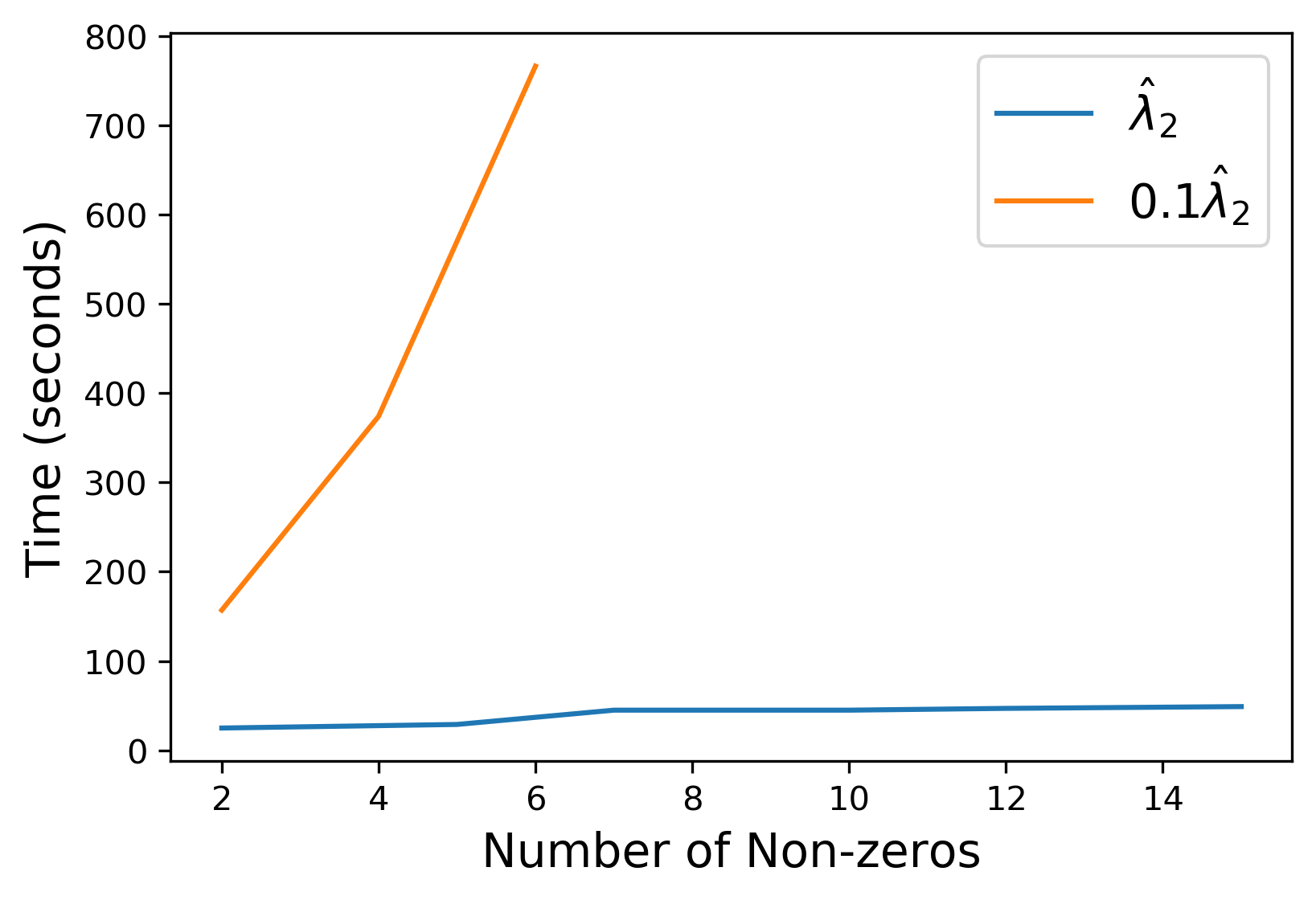}
			\caption{\small{L0BnB run time on a real genomics dataset (Riboflavin) with $n =71$ and  $p \approx 8.3 \times 10^6$.}}
			\label{fig:real}
		\end{minipage}
	\end{figure}

	\section{Conclusion}
	
	We considered the exact computation of estimators from the $\ell_0 \ell_2$-regularized least squares problem. While current approaches for this problem rely on commercial MIP-solvers, we propose a highly specialized nonlinear BnB framework for solving the problem.
	A key workhorse in our approach is a fast coordinate descent procedure to solve the node relaxations along with active set updates and gradient screening, which exploit information across the search tree for computational efficiency. Moreover, we proposed a new method for obtaining dual bounds from our primal coordinate descent solutions and showed that the quality of these bounds depend on the sparsity level, rather than the number of features $p$. Our experiments on both real and synthetic data indicate that our method exhibits over $5000$x speedups compared to the fastest solvers, handling high-dimensional instances with $p = 8.3 \times 10^6$ in the order of seconds to few minutes. Our method appears to be more robust to the choices of the regularization parameters and can handle difficult statistical problems (e.g., relatively high correlations or small number of samples). 
	
	
	Our work demonstrates for the first time that carefully designed first-order methods can be highly effective within a BnB framework; and can perhaps, be applied to more general mixed integer programs involving sparsity. {There are multiple directions for future work. \cite{atamturk20a} recently showed that safe screening rules, which eliminate variables from the optimization problem, can be a very effective preprocessing step for $\ell_0\ell_2$ regularized regression. One promising direction is to extend such rules to dynamically eliminate variables during the course of BnB. Another important direction is to develop specialized methods that can dynamically infer and tighten the Big-M values used in our formulation.}
	
	\subsection*{Acknowledgements}
	Hussein Hazimeh acknowledges research support from the Office of Naval Research ONR-N000141812298. Rahul Mazumder acknowledges research funding from the Office of Naval Research ONR-N000141812298 (Young Investigator Award), the National Science Foundation (NSF-IIS-1718258) and IBM. 
	The authors would like to thank the referees for their constructive comments, which led to an improvement in the paper.

	\bibliographystyle{spmpsci}      
	\bibliography{ref}   
	
	\appendix
	\section{Proofs of Technical Results}
	\textbf{Proof of Theorem \ref{theorem:relaxation}: }
	The interval relaxation of \eqref{eq:conicbigM} can be expressed as:
	\begin{subequations}\label{eq:relax_intermediate}
		\begin{align} 
		\min_{\beta} ~~~ & \Big \{ \frac{1}{2} \| y - X \beta \|_2^2 +  \sum_{i \in [p]} {\min_{s_i,z_i} (\lambda_0 z_i + \lambda_2 s_i)} \Big \}  \\
		& \beta_i^2 \leq s_i z_i, ~  i \in [p] \label{eq:secondcone} \\
		& -{M} z_i \leq \beta_i \leq {M} z_i, ~ i \in [p] \label{eq:M}\\
		& z_i \in [0,1],  s_i \geq 0,  ~ i \in [p] \label{eq:szvars}
		\end{align}
	\end{subequations}
	Let $\omega(\beta_i; \lambda_0, \lambda_2, M) := \min_{s_i,z_i} (\lambda_0 z_i + \lambda_2 s_i) ~ \text{s.t.} ~ \eqref{eq:secondcone}, \eqref{eq:M}, \eqref{eq:szvars}.$ Next we obtain an expression for 
	$\omega(\beta_i; \lambda_0, \lambda_2, M)$.\\
	{Let $(\beta_i, z_i, s_i)$ be a feasible solution for~\eqref{eq:relax_intermediate}. Note that $\hat{z}_i := \max \{ \frac{\beta_i^2}{s_i} , \frac{|\beta_i|}{{M}} \} $ is the smallest possible value of $z_i$, which satisfies constraints \eqref{eq:secondcone} and \eqref{eq:M} (for the case of $\beta_i = s_i = 0$, we define $\beta_i^2/s_i = 0)$. Thus, the objective value corresponding to $(\beta_i, \hat{z}_i, s_i)$ is less than or equal to that of a feasible solution $(\beta_i, z_i, s_i)$. This implies that we can replace constraints \eqref{eq:secondcone} and \eqref{eq:M} with the constraint $z_i = \max \{ \frac{\beta_i^2}{s_i} , \frac{|\beta_i|}{{M}} \}$ without changing the optimal objective of the problem. This replacement leads to:
		$$
		\omega(\beta_i; \lambda_0, \lambda_2, M) = \min_{s_i,z_i}~ (\lambda_0 z_i + \lambda_2 s_i) ~~~ \text{s.t.} ~~~ z_i = \max \Big\{ \frac{\beta_i^2}{s_i} , \frac{|\beta_i|}{{M}} \Big\}, ~ z_i \in [0,1], ~  s_i \geq 0.
		$$
		In the above, we can eliminate the variable $z_i$, leading to: 
		\begin{align}\label{eq:pintermediate}
		\omega(\beta_i; \lambda_0, \lambda_2, M) = \min_{s_i}  &  \max \Big \{ \underbrace{ \lambda_0  \frac{\beta_i^2}{s_i} + \lambda_2 s_i }_{\text{Case I}}, \underbrace{ \lambda_0  \frac{|\beta_i|}{{M}} + \lambda_2 s_i}_{\text{Case II}} \Big \} ~~ \text{s.t.} ~~ s_i \geq \beta_i^2, ~ |\beta_i| \leq M.
		\end{align}}
	Suppose that Case I attains the maximum in \eqref{eq:pintermediate}. This holds if $s_i \leq |\beta_i|M$. The function $\lambda_0  \frac{\beta_i^2}{s_i} + \lambda_2 s_i $ is convex in $s_i$, and  the optimality conditions of this function
	imply that the optimal solution is $s_i^{*} = |\beta_i| \sqrt{{\lambda_0}/{\lambda_2}}$ if $|\beta_i| \leq  \sqrt{{\lambda_0}/{\lambda_2}} \leq M$ and  $s_i^{*} = \beta_i^2$ if $\sqrt{{\lambda_0}/{\lambda_2}} \leq |\beta_i| \leq M$. Plugging $s_i^{*}$ into~\eqref{eq:pintermediate} leads to $\omega(\beta_i; \lambda_0, \lambda_2, M) = 2 \lambda_0 \mathcal{B}( \beta_i \sqrt{{\lambda_2}/{\lambda_0}})$, assuming $\sqrt{{\lambda_0}/{\lambda_2}} \leq  M$.

	Now suppose Case II attains the maximum in \eqref{eq:pintermediate}. This holds if $s_i \geq |\beta_i|M$. The function $\lambda_0  \frac{|\beta_i|}{M} + \lambda_2 s_i$ is monotonically increasing in $s_i$, where $s_i$ is lower bounded by $|\beta_i|M$ and $\beta_i^2$. But we always have $|\beta_i|M \geq \beta_i^2$ (since $|\beta_i| \leq M$), which implies that the optimal solution is $s_i^{*} = |\beta_i|M$. Substituting $s_i^{*}$ into \eqref{eq:pintermediate} we get $\omega(\beta_i; \lambda_0, \lambda_2, M) = (\lambda_0/M + \lambda_2 M) |\beta_i|$, and this holds as long as $\sqrt{{\lambda_0}/{\lambda_2}} >  M$.
	\\ \\
	\textbf{Proof of Proposition \ref{prop:v1v2}: } 
	First, we show~\eqref{eq:v14}. Note that $V_{\text{B}(M)}$ can be simplified to:
	\begin{align} \label{eq:bigm_relax_temp}
	V_{\text{B}(M)} = \min_{ \|\beta\|_{\infty} \leq M} H(\beta) := \frac{1}{2} \| y - X \beta \|_2^2 +  \sum_{i \in [p]} \Big( \frac{\lambda_0}{M} |\beta_i| + \lambda_2 \beta_i^2 \Big).
	\end{align}
	Recalling Theorem~\ref{theorem:relaxation} and the definition of $F(\beta)$ in~\eqref{eq:relaxation}, note that for $\sqrt{{\lambda_0}/{\lambda_2}} > M$:
	$$
	\begin{aligned}
	V_{\text{PR}(M)} - V_{\text{B}(M)} \geq F(\beta^{*}) - H(\beta^{*}) =& \sum\nolimits_{i \in [p]} \Big\{ \psi_2(\beta_i^{*}; \lambda_0, \lambda_2,M) - \frac{\lambda_0}{M} |\beta_i^{*}| - \lambda_2 (\beta_i^{*})^2 \Big\} \\
	=& \lambda_2 \sum\nolimits_{i \in [p]} (M|\beta_i^{*}| - (\beta_i^{*})^2)
	\end{aligned}
	$$
	where the inequality holds since $\beta^{*}$, an optimal solution to~\eqref{eq:relaxation}, is feasible for~\eqref{eq:bigm_relax_temp}. This establishes \eqref{eq:v14}. 
	
	\smallskip
	
	We now show~\eqref{eq:v123}. {Since $|\beta_i^{*}| \leq M$ and $\sqrt{{\lambda_0}/{\lambda_2}} > M$, we must have $\psi_1(\beta_i^{*}; \lambda_0, \lambda_2) = 2 |\beta_i^{*}| \sqrt{\lambda_0 \lambda_2}$ (this corresponds to the first case in \eqref{eq:reverse_huber}). Also recall that $V_{\text{PR}(\infty)} = \min_{\beta} G(\beta)$, where $G(\beta)$ is defined in \eqref{eq:conic_relaxation} (this follows from applying Theorem \ref{theorem:relaxation} with $M = \infty$)}. The following then holds:
	\begin{equation}\label{append-dummy1}
	\begin{aligned}
	V_{\text{PR}(M)} - V_{\text{PR}(\infty)} \geq F(\beta^{*}) - G(\beta^{*}) =& \sum\nolimits_{i \in [p]} \Big\{\psi_2(\beta_i^{*}; \lambda_0, \lambda_2,M) - \psi_1(\beta_i^{*}; \lambda_0, \lambda_2) \Big\}\\
	=& h(\lambda_0, \lambda_2, M) \| \beta^{*} \|_1,
	\end{aligned}
	\end{equation}
	where the inequality holds since $\beta^{*}$, an optimal solution to~\eqref{eq:relaxation}, is feasible 
	for~\eqref{eq:conic_relaxation}. 
	\\ \\ 
	\textbf{Proof of Proposition \ref{prop:relax_obj}: }
	First, we recall that $V_{\text{PR}(\infty)} = \min_{\beta} G(\beta)$, where $G(\beta)$ is defined in \eqref{eq:conic_relaxation}, and that  $V_{\text{B}(M)}$ is defined in \eqref{eq:bigm_relax_temp}. Define a function $t: \mathbb{R} \to \mathbb{R}$ as follows 
	$$t(\beta_i) := 2 \lambda_0 \mathcal{B}( \beta_i \sqrt{{\lambda_2}/{\lambda_0}}) - \frac{\lambda_0}{M} |\beta_i| - \lambda_2 \beta_i^2.$$ Next, we prove~\eqref{eq:relax_obj_1}.
	
	\noindent\textbf{Proof of \eqref{eq:relax_obj_1}:} 
	Suppose that $M \leq \frac{1}{2}\sqrt{{\lambda_0}/{\lambda_2}}$ and  $|\beta_i| \leq M$. Using the fact that $|\beta_i| \leq \sqrt{{\lambda_0}/{\lambda_2}}$ and the definition of $\mathcal{B}$, we have $\mathcal{B}( \beta_i \sqrt{{\lambda_2}/{\lambda_0}}) =  \sqrt{{\lambda_2}/{\lambda_0}} |\beta_i|$. This leads to:
	\begin{align} \label{eq:t_first_part}
	t(\beta_i) = \Big(2 \sqrt{\lambda_0 \lambda_2} - \tfrac{\lambda_0}{M} \Big) |\beta_i| - \lambda_2 \beta_i^2 \leq 0,
	\end{align}
	where the inequality follows since $2 \sqrt{\lambda_0 \lambda_2} - {\lambda_0}/{M} \leq 0$ for $M \leq \frac{1}{2}\sqrt{{\lambda_0}/{\lambda_2}}$. Now, let $\beta^{\dagger}$ be an optimal solution to \eqref{eq:bigm_relax_temp}. Then, the following holds:
	\begin{align} \label{eq:v13_part1}
	V_{\text{B}(M)} - V_{\text{PR}(\infty)} \geq H(\beta^{\dagger}) - G(\beta^{\dagger}) = - \sum\nolimits_{i \in [p]} t(\beta_i^{\dagger}) \geq 0,
	\end{align}
	where the first inequality follows from $V_{\text{B}(M)} = H(\beta^{\dagger})$ and $V_{\text{PR}(\infty)} \leq G(\beta^{\dagger})$. The second inequality in~\eqref{eq:v13_part1} follows from \eqref{eq:t_first_part}. This  establishes \eqref{eq:relax_obj_1}.
	
	\noindent To show \eqref{eq:relax_obj_3}, we will need the following lemma. 
	\begin{lemma} \label{lemma:intermediate_proof_prop_1}
		Let $M \geq \sqrt{{\lambda_0}/{\lambda_2}}$, then for any $\beta_{i} \in [-M, M]$, we have $t(\beta_i) \geq 0$.
		\begin{proof}[Proof of Lemma \ref{lemma:intermediate_proof_prop_1}]
			Suppose that $M \geq \sqrt{{\lambda_0}/{\lambda_2}}$ and $|\beta_i| \leq M$. There are two cases to consider here: 
			Case~(i):~$|\beta_i| \leq \sqrt{{\lambda_0}/{\lambda_2}}$ and 
			Case~(ii):~$|\beta_i| \geq \sqrt{{\lambda_0}/{\lambda_2}}$. \\
			For Case~(i), from the definition of $\mathcal{B}$, we have $2 \lambda_0 \mathcal{B}( \beta_i \sqrt{{\lambda_2}/{\lambda_0}}) = 2 \sqrt{\lambda_0 \lambda_2} |\beta_i|$. Therefore, 
			$$t(\beta_i) = \Big(2 \sqrt{\lambda_0 \lambda_2} - \frac{\lambda_0}{M} \Big) |\beta_i| - \lambda_2 \beta_i^2.$$
			In the above, it is easy to check that for $M \geq \sqrt{{\lambda_0}/{\lambda_2}}$ and $|\beta_i| \leq \sqrt{{\lambda_0}/{\lambda_2}}$, we have $t(\beta_i) \geq 0$. Now for Case (ii), we have $2 \lambda_0 \mathcal{B}( \beta_i \sqrt{{\lambda_2}/{\lambda_0}}) = \lambda_0 + \lambda_2 \beta_i^2$---this leads to 
			$t(\beta_i) = \lambda_0 \Big( 1 - \frac{|\beta_i|}{M} \Big)$, which is 
			non-negative since we assume that $|\beta_i| \leq M$. 
			This establishes Lemma \ref{lemma:intermediate_proof_prop_1}. 
		\end{proof}
	\end{lemma}
	
	\noindent\textbf{Proof of \eqref{eq:relax_obj_3}:}
	Define $v^{*}(M) =  \min_{\| \beta \|_{\infty} \leq M} G(\beta)$ (recall, $G(\beta)$ is defined in \eqref{eq:conic_relaxation}) and let $\beta^{*}$ be an optimal solution i.e., $v^{*}(M) = G(\beta^*)$.
	Suppose that $M \geq \sqrt{{\lambda_0}/{\lambda_2}}$. Then, the following holds:
	\begin{align} \label{eq:relax_obj_2}
	v^{*}(M) - V_{\text{B}(M)} \geq G(\beta^{*}) - H(\beta^{*}) = \sum_{i \in [p]} t(\beta_i^{*}) \geq 0
	\end{align}
	where the first inequality holds since $\beta^{*}$ is a feasible solution to \eqref{eq:bigm_relax_temp} so $V_{\text{B}(M)} \leq H(\beta^{*})$, and the last inequality is due to Lemma \ref{lemma:intermediate_proof_prop_1}. 
	
	Now suppose that $\lambda_2 \in \mathcal{L}(M)$ (defined in~\eqref{eq:lambda2star}) and let $\hat{\beta} \in \mathcal{S}(\lambda_2)$ (i.e., $\hat{\beta}$ is optimal for \eqref{eq:conic_relaxation}) be such that it satisfies $\| \hat{\beta} \|_{\infty} \leq M$. Since $V_{\text{PR}(\infty)} \leq v^{*}(M)$ and $\hat{\beta}$ is feasible for the problem corresponding to $v^{*}(M)$, then $v^{*}(M) = G(\hat{\beta})$. But by \eqref{eq:relax_obj_2} we have $v^{*}(M) \geq V_{\text{B}(M)}$, which combined with $v^{*}(M) = G(\hat{\beta})$, leads to $G(\hat{\beta}) \geq V_{\text{B}(M)}$. This establishes \eqref{eq:relax_obj_3} (since by definition $G(\hat{\beta}) = V_{\text{PR}(\infty)}$); and completes the proof of Proposition \ref{prop:relax_obj}.
	\\ \\
	\textbf{Proof of Proposition \ref{prop:V_explicit}: }
	Fix some $i \in \text{Supp}(\hat{\beta})^c$. Recall that the one-dimensional problem: $\min_{|\beta_i| \leq M} F(\hat{\beta}_1, \dots, \beta_i, \dots, \hat{\beta}_p)$ is equivalent to Problem \eqref{eq:proximal}, where $\tilde{\beta}_i = \langle y - \sum_{j \neq i} X_j \hat{\beta}_j, X_i \rangle$. Since $\hat{\beta}_i = 0$, we have  $\tilde{\beta}_i =  \langle \hat{r}, X_i \rangle$ where, $\hat{r} = y - X\hat{\beta}.$
	For $\sqrt{{\lambda_0}/{\lambda_2}} \leq M$, \eqref{eq:case1sol} implies that the solution of \eqref{eq:proximal} is nonzero iff $|\tilde{\beta}_i| > 2 \sqrt{\lambda_0 \lambda_2}$. Similarly, for $\sqrt{{\lambda_0}/{\lambda_2}} > M$, by \eqref{eq:case2sol}, the solution of \eqref{eq:proximal} is nonzero iff $|\tilde{\beta}_i| > \lambda_0/M + \lambda_2 M$. Using these observations in the definition of $\mathcal{V}$ in Algorithm 2, leads to the result of the proposition.
	\\ \\
	\textbf{Proof of Proposition \ref{prop:screening}: }
	Fix some $i \in \mathcal{V}$. Note that 
	$$
	| \langle \hat{r}, X_i \rangle - \langle {r}^{0}, X_i \rangle | = | \langle X {\beta}^{0} - X \hat{\beta}, X_i \rangle | \leq \| X {\beta}^{0} - X \hat{\beta} \|_2 \| X_i \|_2 \leq \epsilon.
	$$
	Using the triangle inequality and the bound above, we get:
	\begin{align*}
	|\langle \hat{r}, X_i \rangle| & \leq  |\langle {r}^{0}, X_i \rangle | + |  \langle \hat{r}, X_i \rangle - \langle {r}^{0}, X_i \rangle  | \leq |\langle {r}^{0}, X_i \rangle | + \epsilon.
	\end{align*} 
	Therefore, if $i \in \mathcal{V}$, i.e., $| \langle \hat{r}, X_i \rangle| > c(\lambda_0, \lambda_2, M)$, then $|\langle {r}^{0}, X_i \rangle   | +  \epsilon > c(\lambda_0, \lambda_2, M)$, implying that $i \in \hat{\mathcal{V}}$, as defined in~\eqref{eq:V_hat}.
	\\ \\
	\textbf{Proof of Theorem \ref{theorem:duals}: } 
	Problem \eqref{eq:relaxation} can be written equivalently as:
	\begin{align} \label{eq:primalrewrite}
	\min_{\beta \in \mathbb{R}^{p}, r \in \mathbb{R}^n} ~~ & \frac{1}{2} \| r \|_2^2 + \sum_{i\in [p]} \psi(\beta_i; \lambda_0, \lambda_2, M) ~~ \text{s.t.} ~~  r = y - X \beta, ~~ | \beta_i | \leq M, \ \forall i \in [p].
	\end{align}
	
	{\bf{Case of $\sqrt{{\lambda_0}/{\lambda_2}} \leq  M$}:}  We first consider the case when 
	$\sqrt{{\lambda_0}/{\lambda_2}} \leq  M$. We dualize all the constraints in~\eqref{eq:primalrewrite}, leading to the following Lagrangian dual:
	\begin{align} \label{eq:dualintermediate}
	\max_{\alpha \in \mathbb{R}^n, \eta \in \mathbb{R}^p_{\geq 0}} ~~ \min_{\beta \in \mathbb{R}^p, r \in \mathbb{R}^n} ~~ L(\beta, r, \alpha, \eta)
	\end{align}
	where $L(\beta, r, \alpha, \eta)$ is the Lagrangian function defined as follows:
	$$
	L(\beta, r, \alpha, \eta) := \frac{1}{2} \| r \|_2^2 + \sum_{i \in [p]} \psi_1(\beta_i; \lambda_0, \lambda_2) + \alpha^T(r - y + X \beta) + \sum_{i \in [p]} \eta_i (| \beta_i | - M).
	$$
	Next, we discuss how to solve the inner minimization problem in \eqref{eq:dualintermediate}. Let us write:
	\begin{align} \label{eq:inner_min_L}
	\min_{\beta, r} ~~ L(\beta, r, \alpha, \eta) =\min_{r} \Big \{ \frac{1}{2} \| r \|_2^2 + \alpha^T r \Big \} + 
	\sum_{i \in [p]} \min_{\beta_i} \Big \{ D_i(\beta_i) \Big \} - \alpha^T y -  \sum_{i \in [p]} M \eta_i 
	\end{align}
	where $D_i(\beta_i) := \psi_1(\beta_i; \lambda_0, \lambda_2) + \alpha^T X_i \beta_i + \eta_i |\beta_i|$. 
	Note that the minimizer wrt $r$ is given by $\tilde{r}^*=-\alpha$.
	In what follows, we consider some $i \in [p]$ and derive an optimal solution $\tilde{\beta}_i^{*}$ of $\min_{\beta_i} D_i(\beta_i)$. There are two cases to consider here.
	
	\noindent \underline{Case 1:} $|\beta_i| \leq \sqrt{\lambda_0/\lambda_2}$. Here, $\psi_1(\beta_i; \lambda_0, \lambda_2) = 2 \sqrt{\lambda_0 \lambda_2} |\beta_i|$. Note that 
	\begin{equation}\label{thm2-case11}
	\tilde{\beta}^*_{i} = \argmin_{\beta_i} ~ D_i(\beta_i) = \begin{cases}
	0 & \text{if~~~$2 \sqrt{\lambda_0 \lambda_2} + \eta_i - | \alpha^T X_i | \geq 0$}\\
	-\sqrt{\lambda_0/\lambda_2}\sign(\alpha^TX_{i})& \text{otherwise}
	\end{cases}
	\end{equation}
	and as we will see, the second case above is a special case of Case 2 below. 
	
	\noindent \underline{Case 2:}  $|\beta_i| \geq \sqrt{\lambda_0/\lambda_2}$. In this case, $\psi_1(\beta_i; \lambda_0, \lambda_2) = \lambda_2 \beta_i^2 + \lambda_0$. 
	Note that 
	\begin{equation}\label{thm2-case12}
	\tilde{\beta}_i^{*}=\argmin_{\beta_i} D_i(\beta_i) = \argmin_{\beta_i} \lambda_2 \beta_i^2  + \alpha^T X_i \beta_i + \eta_i |\beta_i|
	= {T({-\alpha^T X_i}; {\eta_i}, \infty)}/{(2\lambda_2)}
	\end{equation}
	where, $T$ above is the soft-thresholding operator~\cite{glmnet}.
	But $\tilde{\beta}_i^{*}$ in~\eqref{thm2-case12} should satisfy $|\tilde{\beta}_i^{*}| \geq \sqrt{\lambda_0/\lambda_2}$ in this case. Using the definition of $T(.)$,  the latter condition can be written as: $(| \alpha^T X_i | - \eta_i)/(2\lambda_2) \geq \sqrt{\lambda_0/\lambda_2}$, which simplifies to $| \alpha^T X_i | - \eta_i \geq 2 \sqrt{\lambda_0 \lambda_2}$. 
	In this case, $D_i(\tilde{\beta}_i^{*})$ can be written as:
	\begin{align} \label{eq:D_2}
	D_i(\tilde{\beta}_i^{*}) = 
	- \frac{1}{4 \lambda_2} \Big( |\alpha^T X_i| - \eta_i \Big)^2 + \lambda_0 ~~ \text{ if } ~~ | \alpha^T X_i | - \eta_i \geq 2 \sqrt{\lambda_0 \lambda_2}.
	\end{align}
	Combining~\eqref{thm2-case11} and \eqref{eq:D_2}, we have:
	$D_i(\tilde{\beta}_i^{*}) = - \Big[   \frac{(|\alpha^T X_i| -  \eta_i)^2}{4 \lambda_2} - \lambda_0  \Big]_{+}.$
	Plugging the above expression of $D_i(\tilde{\beta}_i^{*})$ along with $\tilde{r}^{*} = - \alpha$ in \eqref{eq:dualintermediate}, we get the following dual problem: 
	\begin{align}\label{thm2-append-line1}
	\max_{\alpha \in \mathbb{R}^n, \eta \in \mathbb{R}^p_{\geq 0}} & - \frac{1}{2} \| \alpha \|_2^2 - \alpha^T y - \sum_{i=1}^{p} \Big[   \frac{(|\alpha^T X_i| -  \eta_i)^2}{4 \lambda_2} - \lambda_0  \Big]_{+} - \sum_{i=1}^{p} M \eta_i.
	\end{align}

	Note we can drop the non-negativity constraint $\eta_{i} \geq 0$ above. Using a new variable $\gamma$ in place of $\eta$ (note that $\eta_{i} = |\gamma_{i}|$ for all $i$), Problem~\eqref{thm2-append-line1} can be reformulated as:
	\begin{align}\label{thm2-append-line11}
	\max_{\alpha \in \mathbb{R}^n, \gamma \in \mathbb{R}^p } & - \frac{1}{2} \| \alpha \|_2^2 - \alpha^T y - \sum_{i=1}^{p} \Big[   \frac{(\alpha^T X_i -  \gamma_i)^2}{4 \lambda_2} - \lambda_0  \Big]_{+} - \sum_{i=1}^{p} M |\gamma_i|
	\end{align}
	which is the formulation in \eqref{eq:objdualcase1}. 
	
	
	We now show~\eqref{thm2-statement-pr-dual1}. 
	Let $(\alpha^{*}, \eta^{*})$ be an optimal solution of~\eqref{thm2-append-line1}. First, it is easy to see $r^{*} = \argmin_{r} L(\beta^{*}, r, \alpha^{*}, \eta^{*}) = - \alpha^{*}$. 
	For the second part of~\eqref{thm2-statement-pr-dual1}, note by complementary slackness: if $|\beta_i^{*}| < M$ then $\eta_i^{*} = 0$ and consequently $\gamma_{i}^* = 0$.
	Next, we will derive $\gamma_i^{*}$ for the case of $|\beta_i^{*}| = M >0$. We first consider the case of $\beta^{*}_i = M$. Using \eqref{thm2-case12} with the identifications: $\tilde{\beta}^{*}_i = M$, $\alpha = \alpha^{*}$, and $\eta = \eta^{*}$, we get $M = (- \alpha^{*T} X_i - \eta^{*}_i \text{sign}( -  \alpha^{*T} X_i))/(2\lambda_2)$, where $\text{sign}( \alpha^{*T} X_i) = -1$. Using this along with the fact that $\gamma_i^{*} = \eta_i^{*} \text{sign}( \alpha^{*T} X_i)$, we obtain 
	$\gamma_i^{*} =  \alpha^{*T} X_i +   2\lambda_2 M= \alpha^{*T} X_i -   2\lambda_2 M\text{sign}( \alpha^{*T} X_i)$. Using this identity along with a similar argument for $\beta_i = -M$, we arrive at the expression in~\eqref{thm2-statement-pr-dual1}.

	{\bf{Case of $\sqrt{{\lambda_0}/{\lambda_2}} >  M$}:} The proof for this case follows along the lines similar to what was shown above. We omit the proof. 
	\\ \\
	The following lemma is useful for proving Theorem \ref{theorem:dualitygap}.
	\begin{lemma} \label{lemma:v}
		Suppose $\sqrt{{\lambda_0}/{\lambda_2}} \leq  M $. Let $\hat{\beta}$ be a solution from Algorithm~2 and $(\hat{\alpha}, \hat{\gamma})$ be the corresponding dual solution defined in \eqref{eq:case1dualsol}.
		Let $\beta^{*}$ and $r^{*}$ be as defined in Theorem \ref{theorem:duals}, and define the primal gap $\epsilon = \| X (\beta^{*} - \hat{\beta}) \|_2$.  Then, for every $i \in \text{Supp}(\hat{\beta})^c$, we have $v(\hat{\alpha}, \hat{\gamma}_i) = 0$ (see~\eqref{thm2-defn-v1} for definition of $v$); and for every $i \in \text{Supp}(\hat{\beta})$, we have 
		\begin{align} \label{eq:vupbd}
		v(\hat{\alpha}, \hat{\gamma}_i) \leq c_i \epsilon + (4 \lambda_2)^{-1} {\epsilon^2} + v(\alpha^{*}, \gamma^{*}_i),
		\end{align}
		where $c_i = (2\lambda_2)^{-1}$ if $|\beta_i^{*}| < M$, and $c_i = M$ if $|\beta_i^{*}| = M$.
		\begin{proof}[Proof of Lemma \ref{lemma:v}]
			Fix some $i \in \text{Supp}(\hat{\beta})^c$. Since $\hat{\beta}$ is the output of Algorithm 2, the set $\mathcal{V}$ in Algorithm 2 must be empty. Thus, using Proposition \ref{prop:V_explicit}, we have: $|\hat{r}^{T} X_i| \leq 2 \sqrt{\lambda_0 \lambda_2}$.
			As $\hat{r} = -\hat{\alpha}$, and consequently using the definition of $\hat{\gamma}$ in~\eqref{eq:case1dualsolexplicit}, we have
			$\hat{\gamma}_i =0$. Thus, 
			$${(\hat{\alpha}^T X_i - \hat{\gamma}_i)^2}/{(4 \lambda_2)} =  {(\hat{r}^T X_i)^2}/{(4 \lambda_2)} \leq  \lambda_0,$$
			which implies that $v(\hat{\alpha}, \hat{\gamma}_i) = 0$.
			
			Now fix some $i \in \text{Supp}(\hat{\beta})$, and let 
			$a := (\hat{\alpha} - \alpha^{*})^T X_i$ and $b := \alpha^{*T} X_i - \gamma^{*}_i.$ By \eqref{eq:case1dualsol} and the definition of $h_{1}$ in~\eqref{eq:objdualcase1}, we have $\hat{\gamma}_i = \argmin_{\gamma_i} v(\hat{\alpha}, \gamma_i)$, which leads to $v(\hat{\alpha}, \hat{\gamma}_i) \leq v(\hat{\alpha}, \gamma^{*}_i)$. An upper bound on $v(\hat{\alpha}, \hat{\gamma}_i)$ can be then obtained as follows:
			\begin{align}
			v(\hat{\alpha}, \hat{\gamma}_i) \leq v(\hat{\alpha}, \gamma^{*}_i) & = \Big[ \frac{(\hat{\alpha}^T X_i -  \gamma^{*}_i)^2}{4 \lambda_2} - \lambda_0  \Big]_{+} + M |\gamma^{*}_i|  \nonumber \\
			& = \Big[ \frac{((\hat{\alpha} - \alpha^{*})^T X_i + \alpha^{*T} X_i -  \gamma^{*}_i)^2}{4 \lambda_2} - \lambda_0  \Big]_{+} + M |\gamma^{*}_i| \nonumber \\ & = \Big[ \frac{a^2 + 2ab + b^2}{4 \lambda_2} - \lambda_0  \Big]_{+} + M |\gamma^{*}_i| \nonumber \\
			& \leq \frac{a^2 + 2|a| |b|}{4 \lambda_2} +  \Big[ \frac{b^2}{4 \lambda_2} - \lambda_0  \Big]_{+} + M |\gamma^{*}_i| \nonumber \\ 
			& \leq \frac{a^2 + 2|a| |b|}{4 \lambda_2} + v(\alpha^{*}, \gamma^{*}_i)\label{eq:bdv}.
			\end{align}
			Next, we obtain upper bounds on $|a|$ and $|b|$. 
			By the Cauchy-Schwarz inequality, we have $|a| \leq \| X (\beta^* - \hat{\beta})\|_{2} \| X_{i} \|_{2} = \epsilon \| X_i \|_2$. Since the columns of $X$ are normalized, the bound simplifies to $|a| \leq \epsilon$. 
			Since, $\| y \|_{2}=1$ and $\alpha^{*} = - r^{*}$ (see Theorem \ref{theorem:duals}), we have:
			$$\frac{1}{2} \| \alpha^{*} \|_2^2 = \frac{1}{2} \| r^{*} \|_2^2 \leq  F(\beta^{*}) \leq F(0) = \frac{1}{2} \| y \|_2^2 = \frac{1}{2},$$
			implying that $\| \alpha^{*} \|_2 \leq 1$.
			
			\noindent We now present a bound on $|b| = |\alpha^{*T} X_i - \gamma^{*}_i|$ by considering two cases: 
			\begin{itemize}
				\item[] Case 1: ($|\beta_i^{*}| < M$): From the definition of $\gamma_i^{*}$ in \eqref{thm2-statement-pr-dual1}, we have $\gamma^{*}_i = 0$, which leads to $|b| = |\alpha^{*T} X_i| \leq \| \alpha^{*} \|_2 \| X_{i} \|_2 \leq 1$. 
				\item[] Case 2: ($|\beta_i^{*}| = M$): By \eqref{thm2-statement-pr-dual1},  $\gamma^{*}_i = {\alpha}^{*T} X_i -  2M\lambda_2 \sign({\alpha}^{*T} X_i)$, which leads to $ |b| \leq  2M\lambda_2$. 
			\end{itemize}
			Plugging $|a| \leq \epsilon$ and the bounds on $|b|$ (above) into \eqref{eq:bdv}, we arrive to \eqref{eq:vupbd}.
		\end{proof}
	\end{lemma}
	\noindent \textbf{Proof of Theorem \ref{theorem:dualitygap}: } We consider two cases based on $\sqrt{{\lambda_0}/{\lambda_2}} \leq  M $ or $\sqrt{{\lambda_0}/{\lambda_2}} >  M $. \\
	\noindent {\bf Showing bound~\eqref{eq:dualbd1}:} First, we consider the case of $\sqrt{{\lambda_0}/{\lambda_2}} \leq  M $, i.e., we will establish the bound in \eqref{eq:dualbd1}. Note that the following holds:
	\begin{align}
	\frac{1}{2} \| \hat{\alpha} \|_2^2 + \hat{\alpha}^{T} y & = \frac{1}{2} \| \hat{\alpha} - \alpha^{*} + \alpha^{*}  \|_2^2 + (\hat{\alpha} - \alpha^{*} + \alpha^{*})^{T} y \nonumber \\
	& = \frac{1}{2} \| \hat{\alpha} - \alpha^{*} \|_2^2 + (\hat{\alpha} - \alpha^{*})^{T} \alpha^{*} + \frac{1}{2} \| \alpha^{*} \|_2^2 + (\hat{\alpha} - \alpha^{*})^{T} y + \alpha^{*T} y \nonumber \\
	& \leq \frac{1}{2} \epsilon^2 + \epsilon \| \alpha^{*} \|_2 + \frac{1}{2} \| \alpha^{*} \|_2^2  + \epsilon \|y \|_2 + \alpha^{*T} y \label{eq:cauchyintermediate},
	\end{align}
	where the inequality above follows by applying Cauchy-Schwarz and noting that $\|\hat{\alpha} - \alpha^{*} \|_2 = \| X (\beta^{*} - \hat{\beta}) \|_2 = \epsilon$. Using $\| y \|_2 = 1$ and $\| \alpha^{*} \|_2 \leq 1$ (this was established in the proof of Lemma \ref{lemma:v}) in \eqref{eq:cauchyintermediate}, we get:
	\begin{align}
	\frac{1}{2} \| \hat{\alpha} \|_2^2 + \hat{\alpha}^{T} y \leq \frac{1}{2} \epsilon^2 + 2 \epsilon + \frac{1}{2} \| \alpha^{*} \|_2^2 + \alpha^{*T} y \label{eq:yboundedintermediate}.
	\end{align}
	Plugging \eqref{eq:yboundedintermediate} into the objective function in  \eqref{eq:objdualcase1} (with $\alpha = \hat{\alpha}$ and $\gamma = \hat{\gamma}$):
	\begin{align} \label{eq:h1upbdintermediate}
	h_1(\hat{\alpha}, \hat{\gamma}) & \geq -\frac{1}{2} \| \alpha^{*} \|_2^2 -  \alpha^{*T} y - 2 \epsilon - \frac{1}{2} \epsilon^2 
	- \sum_{i \in [p]} v(\hat{\alpha}, \hat{\gamma}_i).
	\end{align}
	By Lemma \ref{lemma:v}, for every $i \in \text{Supp}(\hat{\beta})^c$, we have $v(\hat{\alpha}, \hat{\gamma}_i) = 0$, which implies $v(\hat{\alpha}, \hat{\gamma}_i) \leq v({\alpha}^{*}, {\gamma}^{*}_i)$ (since $v$ is a non-negative function). 
	
	Using the inequality $v(\hat{\alpha}, \hat{\gamma}_i) \leq v({\alpha}^{*}, {\gamma}^{*}_i)$ for every $i \in \text{Supp}(\hat{\beta})^c$; 
	and inequality~\eqref{eq:vupbd} for every $i \in \text{Supp}(\hat{\beta})$, in \eqref{eq:h1upbdintermediate}, we get:
	\begin{align}
	h_1(\hat{\alpha}, \hat{\gamma}) & \geq -\frac{1}{2} \| \alpha^{*} \|_2^2 -  \alpha^{*T} y - \sum_{i \in [p]} v({\alpha}^{*}, {\gamma}_i^{*}) - 2 \epsilon - \frac{1}{2} \epsilon^2 - \sum_{i \in \text{Supp}(\hat{\beta})} \Big( c_i \epsilon + (4 \lambda_2)^{-1} {\epsilon^2} \Big) \nonumber \\
	&= h_1(\alpha^{*}, \gamma^{*}) - 2 \epsilon - \frac{1}{2} \epsilon^2 - \sum_{i \in \text{Supp}(\hat{\beta})} \Big( c_i \epsilon + (4 \lambda_2)^{-1} {\epsilon^2} \Big).  \label{eq:h_1_bd_temp}
	\end{align}
	Let 
	$$k_1 = | \{ i \in \text{Supp}(\hat{\beta}) \ | \ |\hat{\beta}_i| < M \} |~~\text{and} ~~k_2 = | \{ i \in \text{Supp}(\hat{\beta}) \ | \ |\hat{\beta}_i| = M \} |.$$ 
	Using the expressions for $c_i$s (from Lemma \ref{lemma:v}) in~\eqref{eq:h_1_bd_temp}, we get:
	\begin{align*}
	h_1(\hat{\alpha}, \hat{\gamma}) \geq  h_1(\alpha^{*}, \gamma^{*}) - 2 \epsilon - \frac{1}{2} \epsilon^2 - k_1 \Big( (2\lambda_2)^{-1} \epsilon + (4 \lambda_2)^{-1} {\epsilon^2} \Big) - k_2 \Big( M \epsilon + (4 \lambda_2)^{-1} {\epsilon^2} \Big)
	\end{align*}
	Rearranging the terms in the above and using the fact that $k = k_1 + k_2$, we get:
	\begin{align*}
	h_1(\hat{\alpha}, \hat{\gamma}) \geq h_1(\alpha^{*}, \gamma^{*}) - \epsilon (2 + k_1 (2\lambda_2)^{-1} + k_2 M) - \epsilon^2 \Big(\frac{1}{2} + \frac{k}{4 \lambda_2}\Big),
	\end{align*}
	which leads to \eqref{eq:dualbd1}.
	
	\noindent {\bf Showing bound~\eqref{eq:dualbd2}:} Now, we consider the case of $\sqrt{{\lambda_0}/{\lambda_2}} >  M $, where we will establish the bound in \eqref{eq:dualbd2}. By the same argument used in deriving \eqref{eq:h1upbdintermediate}, we have:
	\begin{align} \label{eq:h2upbdintermediate}
	h_2(\hat{\rho}, \hat{\mu}) & \geq -\frac{1}{2} \| \rho^{*} \|_2^2 -  \rho^{*T} y - 2 \epsilon - \frac{1}{2} \epsilon^2 
	- M \| \hat{\mu} \|_1.
	\end{align}
	Next, we fix some $i \in \text{Supp}(\hat{\beta})$ and we upper bound $\hat{\mu}_i$ as follows:
	\begin{align}
	|\hat{\mu}_i|  = \Big [|\hat{\rho}^T X_i| - \lambda_0/M - \lambda_2 M \Big]_{+} \leq&
	\Big [|(\hat{\rho} - \rho^{*})^T X_i| + |\rho^{*T} X_i| - \lambda_0/M - \lambda_2 M \Big]_{+} \nonumber \\
	\leq& |(\hat{\rho} - \rho^{*})^T X_i| + \Big [|\rho^{*T} X_i| - \lambda_0/M - \lambda_2 M \Big]_{+} \nonumber \\
	\leq& \epsilon + |\mu_i^{*}|, \label{eq:gammaupbdcase2}
	\end{align}
	where in the last step above we use Cauchy-Schwarz for the first term (noting that $\rho^*=-r^*$); and for the second term, we use Theorem~\ref{theorem:duals} along with the following:
	\begin{align} \label{eq:conditions_rho}
	|\rho^{*T} X_i| \leq \lambda_0/M + \lambda_2 M~~\text{if}~~ |{\beta}^{*}_i| < M~~\text{and}~~ |\rho^{*T} X_i| \geq \lambda_0/M + \lambda_2 M~~\text{if}~~|{\beta}^{*}_i| = M.
	\end{align}
	Conditions in \eqref{eq:conditions_rho} follow from the optimality of
	$(\rho^{*}, \mu^{*})$ for \eqref{eq:objdualcase2}.
	(For $| \beta_i^{*} | < M$, \eqref{thm2-statement-pr-dual2} implies that $\mu_i^{*} = 0$, which leads to $|\rho^{*T} X_i| \leq \lambda_0/M + \lambda_2 M$ from the feasibility condition in~\eqref{eq:objdualcase2}. When $| \beta_i^{*} | = M$, we will establish~\eqref{eq:conditions_rho} via contradiction: 
	If $|\rho^{*T} X_i| < \lambda_0/M + \lambda_2 M$ holds true, then~\eqref{thm2-statement-pr-dual2} implies $\mu^*_{i} <0$, which would violate optimality for~\eqref{eq:objdualcase2}.) 
	
	For $i \in  \text{Supp}(\hat{\beta})^c$, we have $\hat{\mu}_i = 0$ (see the discussion after \eqref{eq:mu_hat}), which implies $|\hat{\mu}_i| \leq |\mu^{*}_i|$. Using the inequalities $|\hat{\mu}_i| \leq |\mu^{*}_i|$, $i \in  \text{Supp}(\hat{\beta})^c$ and \eqref{eq:gammaupbdcase2} (for all $i \in  \text{Supp}(\hat{\beta})$) in \eqref{eq:h2upbdintermediate}, and simplifying, we get:
	\begin{align}
	h_2(\hat{\rho}, \hat{\mu}) \geq h_2({\rho}^{*}, {\mu}^{*}) - \epsilon (2 + M k) - \epsilon^2/2,
	\end{align}
	which leads to \eqref{eq:dualbd2}.

	\section{Additional Technical Details} \label{sec:additional_techincal_details}
	
	\subsection{Derivation of solutions to Problem \eqref{eq:proximal}}\label{appendix:thresholding}
	{First, we consider the setting of $\sqrt{\lambda_0/\lambda_2} \leq M$. From Theorem \ref{theorem:relaxation}, the penalty $\psi(\beta_i; \lambda_0, \lambda_2, M) = \psi_{1} (\beta_i; \lambda_0, \lambda_2) =  2\lambda_0 \mathcal{B}(\beta_i \sqrt{\lambda_2/\lambda_0})$. Using the definition of $\mathcal{B}$ in \eqref{eq:reverse_huber}, we have:
		\begin{align*}
		\psi_{1} (\beta_i; \lambda_0, \lambda_2) = 
		\begin{cases}
		2 \sqrt{\lambda_0 \lambda_2} |\beta_i|  & ~~ \text{ if } |\beta_i| \leq \sqrt{\lambda_0/\lambda_2} \\
		\lambda_2 \beta_i^2 + \lambda_0 & ~~ \text{ if } |\beta_i| \geq  \sqrt{\lambda_0/\lambda_2}.
		\end{cases}
		\end{align*}}
	
	{\textbf{Case of $|\beta_i| \leq \sqrt{\lambda_0/\lambda_2}$: } The penalty in this case is $2 \sqrt{\lambda_0 \lambda_2} |\beta_i|$, and the first-order optimality conditions imply that the solution of Problem \eqref{eq:proximal} is given by a capped version (due to the box constraint on $\beta_{i}$) of the soft thresholding operator~\cite{glmnet}: $\beta_i^{*} = T(\tilde{\beta}_i; 2 \sqrt{\lambda_0 \lambda_2}, M)$; and this is optimal as long as: $|\beta_i^{*}| \leq \sqrt{\lambda_0/\lambda_2}$, i.e., $|T(\tilde{\beta}_i; 2 \sqrt{\lambda_0 \lambda_2}, M)| \leq \sqrt{\lambda_0/\lambda_2}$.
		Therefore, if $|\tilde{\beta}_i| \leq \sqrt{\lambda_0/\lambda_2} + 2\sqrt{\lambda_0 \lambda_2} $, the solution is given by $T(\tilde{\beta}_i; 2 \sqrt{\lambda_0 \lambda_2}, M)$.}
	
	{\textbf{Case of $|\beta_i| \geq \sqrt{\lambda_0/\lambda_2}$: } Here the penalty is $\lambda_2 \beta_i^2 + \lambda_0$, and the solution is given by $T(\tilde{\beta}_i(1 + 2 \lambda_2)^{-1}; 0, M)$. The latter solution is optimal as long as $|T(\tilde{\beta}_i(1 + 2 \lambda_2)^{-1}; 0, M)| \geq \sqrt{\lambda_0/\lambda_2}$, 
		which can be simplified to 
		$|\tilde{\beta}_i| \geq \sqrt{\lambda_0/\lambda_2} + 2\sqrt{\lambda_0 \lambda_2}$.\\
		This completes the derivation of~\eqref{eq:case1sol}.}
	
	\smallskip
	
	\noindent {Finally, we consider the setting of $\sqrt{{\lambda_0}/{\lambda_2}} >  M$. By Theorem \ref{theorem:relaxation}, the penalty is $\psi(\beta_i; \lambda_0, \lambda_2, M) = (\lambda_0/M + \lambda_2 M)|\beta_i|$. Thus, using arguments similar to the above, the solution is $T(\tilde{\beta}_i; \lambda_0/M + \lambda_2 M, M)$. This completes the derivation of~\eqref{eq:case2sol}.}
	
	\subsection{Node Subproblems} \label{appendix:node_subproblems}
	{In Theorem \ref{theorem:relaxation}, we presented a reformulation of the root relaxation in the $\beta$ space. Here we discuss how to similarly reformulate and solve the subproblem at an arbitrary node in the search tree. Recall that at any node, some $z_i$s can be fixed to 0 or 1 (this depends on the branching decisions made until reaching the node). At a given node, let $\mathcal{Z}$ and $\mathcal{N}$ be the sets of indices of the $z_i$s that are fixed to $0$ and $1$,  respectively. Then, the following convex subproblem needs to be solved at the node:
		\begin{equation} \label{eq:node_subproblem}
		\begin{aligned} 
		\min_{\beta, z, s} ~~&~ \frac{1}{2} \| y - X \beta \|^{2}_2  + \lambda_0 \sum_{i \in [p]} z_i+ \lambda_2 \sum_{i \in [p]} s_i  \\
		\text{s.t.}~~&~~ \beta_i^2 \leq s_i z_i, ~i \in [p]\\
		&~~ -M z_i \leq \beta_i \leq { M} z_i, ~ i \in [p] \\
		&~~ s_i \geq 0, ~  i \in [p]\\
		&~~ z_i = 0, i \in \mathcal{Z}, ~~ z_i = 1, i \in \mathcal{N}, ~~ z_i \in [0,1], i \in [p]\setminus (\mathcal{Z} \cup \mathcal{N}).
		\end{aligned}
		\end{equation}
		Define the penalty $\tilde{\psi}(\beta_i; \lambda_0, \lambda_2) := \lambda_0 + \lambda_2 \beta_i^2$. Then, using an argument similar to that in the proof of Theorem \ref{theorem:relaxation}, Problem \eqref{eq:node_subproblem} can be equivalently expressed in the $\beta$ space as:
		\begin{equation} \label{eq:node_subproblem_reformulated}
		\begin{aligned}
		\min_{\beta \in \mathbb{R}^p} & ~ \tilde{F}(\beta):= \frac{1}{2} \| y - X \beta \|_2^2 + \sum_{i \in \mathcal{N}^{c}} \psi(\beta_i; \lambda_0, \lambda_2, M) + \sum_{i \in \mathcal{N}} \tilde{\psi}(\beta_i; \lambda_0, \lambda_2)  \\ & ~~ \text{s.t.} ~~ \|\beta\|_{\infty} \leq M, ~~ \beta_i = 0, i \in \mathcal{Z}.
		\end{aligned}
		\end{equation}
		Note that Problem \eqref{eq:node_subproblem_reformulated} is similar to Problem \eqref{eq:relaxation}, except that (i) the variables corresponding to $\mathcal{N}$  use the penalty $\tilde{\psi}(\beta_i; \lambda_0, \lambda_2)$ instead of $\psi(\beta_i; \lambda_0, \lambda_2, M)$, and (ii) the variables corresponding to $\mathcal{Z}$ are fixed to zero. To optimize Problem \eqref{eq:node_subproblem_reformulated} using cyclic CD, the coordinates in $(\mathcal{N}\cup \mathcal{Z})^c$ are updated exactly as described in Section \ref{sec:CD}. Next, we discuss how cyclic CD updates the coordinates in $\mathcal{N}$. For any $i \in \mathcal{N}$, the algorithm updates coordinate $i$ by solving the problem:
		$$
		\min_{\beta_i \in \mathbb{R}} \tilde{F}(\hat{\beta}_1, \dots, \beta_i, \dots, \hat{\beta}_p) ~~~~ \text{ s.t. } ~~~ |\beta_i| \leq M.
		$$
		Assuming the columns of $X$ have unit $\ell_2$ norm, the problem above is equivalent to~\eqref{eq:proximal} but with $\psi(\beta_i; \lambda_0, \lambda_2, M)$ replaced with $\tilde{\psi}(\beta_i; \lambda_0, \lambda_2)$; and the  corresponding solution is given by: $T(\tilde{\beta}_i(1 + 2 \lambda_2)^{-1}; 0, M)$. }

	\section{Additional Experimental Results and Details} \label{appendix:additional_results}
	
	\subsection{Experiment of Section \ref{sec:timing_comparison}} \label{appendix:experiment_vary_p_m}
	
	{In Table \ref{table:gradinet_screening}, we report the running time of L0BnB, with and without gradient screening. The number of nodes explored by the different solvers is reported in Table \ref{table:varyp_m_nodes}. In this experiment: $\lambda_2^{*} = 0.0409 $; and $\lambda_0^{*}$ is equal to 0.012, 0.0115, and 0.0132, for $\lambda_{2}$ set to $\lambda_2^{*}$, $0.1 \lambda_2^{*}$, and $10 \lambda_2^{*}$, respectively.}
	
	\begin{table}[htbp]
		\caption{{{\small Running time in seconds for L0BnB, with and without gradient screening (GS). This experiment is based on the same data and parameters as Table \ref{table:fixed_M}. Both approaches explore exactly the same BnB tree.}}}
		\label{table:gradinet_screening}
		\centering
		\renewcommand{\tabcolsep}{2pt}
		\renewcommand{\arraystretch}{1.1}
		\qquad \qquad
		{\begin{tabular}{l|cc|cc|cc|}
				\cline{2-7}
				& \multicolumn{2}{c|}{$\lambda_2 = \lambda_2^{*}$} & \multicolumn{2}{c|}{$\lambda_2 = 0.1 \lambda_2^{*}$} & \multicolumn{2}{c|}{$\lambda_2 = 10 \lambda_2^{*}$} \\ \hline
				\multicolumn{1}{|l|}{p}      & No GS                   & With GS                & No GS                  & With GS                     & No GS                    & With GS                  \\ \hline
				\multicolumn{1}{|l|}{$10^3$} & \textbf{0.7}            & 0.9                    & \textbf{1.6}          & 2.1                           & 0.011             & {0.011}                        \\
				\multicolumn{1}{|l|}{$10^4$} & 3.3                  & \textbf{2.9}               & 11.5                     & \textbf{11.3}                 & 0.06                     & \textbf{0.05}            \\
				\multicolumn{1}{|l|}{$10^5$} & 34                      & \textbf{19}            & 545                    & \textbf{459}                & 0.5                      & 0.5                      \\
				\multicolumn{1}{|l|}{$10^6$} & 1112                    & \textbf{414}           & (23\%)                 & \textbf{(8\%)}                  & 8.1                        & \textbf{7.3}               \\ \hline
			\end{tabular}
			\newline
			\vspace*{0.2cm}
			\newline
			\begin{tabular}{l|cc|cc|cc|cc|}
				\cline{2-9}
				& \multicolumn{2}{c|}{$M=M^{*}$} & \multicolumn{2}{c|}{$M=2 M^{*}$} & \multicolumn{2}{c|}{$M= 4 M^{*}$} & \multicolumn{2}{c|}{$M=\infty$} \\ \hline
				\multicolumn{1}{|l|}{p}      & No GS          & With GS       & No GS           & With GS        & No GS           & With GS         & No GS          & With GS        \\ \hline
				\multicolumn{1}{|l|}{$10^3$} & \textbf{0.16}   & 0.25           & \textbf{0.8}    & 0.98           & \textbf{0.8}    & 1.1             & \textbf{0.8}   & 1.1            \\
				\multicolumn{1}{|l|}{$10^4$} & 0.9            & \textbf{0.6}  & 4.9             & \textbf{4.3}   & 5.4             & \textbf{4.8}    & 5.5            & \textbf{4.8}   \\
				\multicolumn{1}{|l|}{$10^5$} & 7.5            & \textbf{3.9}  & 70              & \textbf{43}    & 81              & \textbf{50}     & 81             & \textbf{50}    \\
				\multicolumn{1}{|l|}{$10^6$} & 121            & \textbf{52}   & 9265            & \textbf{4640}     & 10986           & \textbf{5639}      & 11010          & \textbf{5612}     \\ \hline
		\end{tabular}}
	\end{table}
	
	\begin{table}[htbp]
		\caption{{\small Number of BnB nodes explored in the experiment of Section \ref{sec:timing_comparison}. For~\cite{bertsimas2017sparse}, we report the number of cuts used by the cutting-plane algorithm. A dash indicates that the method could not solve more than 1 node in 4 hours.}}\label{table:varyp_m_nodes}
		\centering
		\renewcommand{\tabcolsep}{1.4pt}
		\renewcommand{\arraystretch}{1.2}
		{\begin{tabular}{l|ccccc|ccccc|ccccc|}
				\cline{2-16}
				& \multicolumn{5}{c|}{$\lambda_2 = \lambda_2^{*}$} & \multicolumn{5}{c|}{$\lambda_2 = 0.1 \lambda_2^{*}$} & \multicolumn{5}{c|}{$\lambda_2 = 10 \lambda_2^{*}$} \\ \hline
				\multicolumn{1}{|l|}{p}      & L0BnB     & GRB       & MSK       & B     & \cite{bertsimas2017sparse}       & L0BnB      & GRB       & MSK       & B      & \cite{bertsimas2017sparse}         & L0BnB      & GRB        & MSK         & B      & \cite{bertsimas2017sparse}     \\ \hline
				\multicolumn{1}{|l|}{$10^3$} & 159       & 532     & 161     & 2     & 9895     & 329        & 31      & 13      & -      & 35563      & 3          & 374      & 347       & -      & 9      \\
				\multicolumn{1}{|l|}{$10^4$} & 225       & 506     & 382     & -     & 4769     & 531       & 3       & 13      & -      & 5212       & 3          & 357      & 1143      & -      & 9      \\
				\multicolumn{1}{|l|}{$10^5$} & 385       & -       & -       & -     & -        & 4337      & -       & -       & -      & -          & 3          & -        & -         & -      & 10     \\
				\multicolumn{1}{|l|}{$10^6$} & 1087      & -       & -       & -     & -        & 10104       & -       & -       & -      & -          & 3          & -        & -         & -      & 11     \\ \hline
			\end{tabular}
			\newline
			\vspace*{0.2cm}
			\newline
			\begin{tabular}{l|cccc|cccc|cccc|cccc|}
				\cline{2-17}
				& \multicolumn{4}{c|}{$M = M^{*}$} & \multicolumn{4}{c|}{$M =2 M^{*}$} & \multicolumn{4}{c|}{$M = 4 M^{*}$} & \multicolumn{4}{c|}{$M = \infty$} \\ \hline
				\multicolumn{1}{|l|}{p}      & L0BnB    & GRB      & MSK     & B    & L0BnB    & GRB      & MSK      & B   & L0BnB    & GRB       & MSK      & B    & L0BnB    & GRB      & MSK       & B   \\ \hline
				\multicolumn{1}{|l|}{$10^3$} & 49       & 65     & 48    & 3    & 185      & 714    & 258    & -   & 193      & 2120    & 261    & -    & 193      & 658    & 273     & -   \\
				\multicolumn{1}{|l|}{$10^4$} & 61       & 128    & 67    & -    & 287      & 2      & 867    & -   & 305      & 1       & 852    & -    & 305      & -      & 1403    & -   \\
				\multicolumn{1}{|l|}{$10^5$} & 75      & -      & -     & -    & 693     & -      & -      & -   & 761     & -       & -      & -    & 761     & -      & -       & -   \\
				\multicolumn{1}{|l|}{$10^6$} & 119      & -      & -     & -    & 8919     & -      & -      & -   & 10059     & -       & -      & -    & 9805     & -      & -       & -   \\ \hline
		\end{tabular}}
	\end{table}

	\subsection{Experiment of Section \ref{sec:vary_snr}} \label{appendix:vary_snr}
	{The parameters $\lambda_0$ and $\lambda_2$, and the number of BnB nodes explored are reported in Table \ref{table:num_nodes_vary_snr}.}
	
	\begin{table}[htbp]
		\caption{{\small The parameters $\lambda_0$ and $\lambda_2$, and the number of nodes explored by the different solvers in the experiment of Section \ref{sec:vary_snr}. The parameter $M$ is reported in the main text.}}
		\label{table:num_nodes_vary_snr}
		\renewcommand{\tabcolsep}{1.4pt}
		\renewcommand{\arraystretch}{1.2}
		\centering
		{\begin{tabular}{|c|c|c|cccc|}
				\hline
				SNR & $\lambda_0^{*}$ & $\lambda_2^{*}$ & L0BnB & GRB & MSK & B \\ \hline
				0.5 & 0.00402 & 0.126     & 34155    & 85402   & 7784  & 1     \\
				1   & 0.0057 & 0.0869   & 223   & 4020   & 2840  & 1     \\
				2   & 0.0097 & 0.0596   & 217   & 868    & 736   & 1     \\
				3   & 0.0113 & 0.0409    & 221   & 461    & 400   & 1     \\
				4   & 0.0121 & 0.0409     & 189   & 428    & 227   & 1     \\
				5   & 0.0120 & 0.0409       & 159   & 532    & 161   & 1     \\ \hline
		\end{tabular}}
	\end{table}
	
	\subsection{Experiment of Section \ref{sec:vary_grid}} \label{appendix:vary_l0_l2}
	{The parameter $\lambda_2^{*} = 0.0409$. The parameter $\lambda_0$ and the number of BnB nodes explored are reported in Table \ref{table:vary_grid_appendix}.}

	\begin{table}[htbp]
		\caption{\small{The parameter $\lambda_0$ and the number of BnB nodes explored in the experiment of Section \ref{sec:vary_grid}.}}
		\label{table:vary_grid_appendix}
		\centering
		\qquad \qquad \qquad
		{\begin{tabular}{|cc|cccc|}
				\hline
				\multicolumn{6}{|c|}{${p=10^3}$} \\
				\hline
				$\lambda_2$                          & $\lambda_0$ & L0BnB  & GRB   & MSK   & B \\ \hline
				\multirow{3}{*}{$10 \lambda_2^{*}$} 
				& 0.02692     & 343    & 60303 & 376   & 1 \\
				& 0.00538     & 1      & 27103 & 1     & 1 \\
				& 0.00054     & 270993 & 2056  & 10271 & 1 \\ \hline
				\multirow{3}{*}{$2 \lambda_2^{*}$}  
				& 0.04207     & 1051   & 55882 & 1764  & 1 \\
				& 0.00841     & 27     & 354   & 89    & 1 \\
				& 0.00084     & 187743 & 2045  & 9977  & 1 \\ \hline
				\multirow{3}{*}{$ \lambda_2^{*}$}   
				& 0.04526     & 2825   & 26721 & 7564  & 1 \\
				& 0.00905     & 71     & 343   & 237   & 1 \\
				& 0.00091     & 175876 & 27259 & 9819  & 1 \\ \hline
				\multirow{3}{*}{$0.5 \lambda_2^{*}$}
				& 0.04704     & 6417   & 18821 & 10959 & 1 \\
				& 0.00941     & 113    & 242   & 394   & 1 \\
				& 0.00094     & 66769  & 18817 & 9884  & 1 \\ \hline
				\multirow{3}{*}{$0.1 \lambda_2^{*}$}
				& 0.04856     & 11811  & 12518 & 13643 & 1 \\
				& 0.00971     & 175    & 218   & 802   & 1 \\
				& 0.00097     & 22096  & 19983 & 10228 & 1 \\ \hline
			\end{tabular}
			\newline
			\vspace*{0.2cm}
			\newline
			\begin{tabular}{|cc|cccc|}
				\hline
				\multicolumn{6}{|c|}{${p=10^4}$} \\
				\hline
				$\lambda_2$                          & $\lambda_0$ & L0BnB & GRB & MSK  & B \\ \hline
				\multirow{3}{*}{$10 \lambda_2^{*}$} 
				& 0.02692     & 343   & 1   & 378  & 1 \\
				& 0.00538     & 1     & 1   & 1    & 1 \\
				& 0.00054     & 56504 & 1   & 826  & 1 \\ \hline
				\multirow{3}{*}{$2 \lambda_2^{*}$}  
				& 0.04207     & 10097 & 1   & 901  & 1 \\
				& 0.00841     & 37    & 1   & 470  & 1 \\
				& 0.00084     & 49358 & 1   & 820  & 1 \\ \hline
				\multirow{3}{*}{$ \lambda_2^{*}$}   
				& 0.04526     & 43745 & 1   & 886  & 1 \\
				& 0.00905     & 111   & 1   & 801  & 1 \\
				& 0.00091     & 32142 & 1   & 824  & 1 \\ \hline
				\multirow{3}{*}{$0.5 \lambda_2^{*}$}
				& 0.04704     & 75893 & 353 & 915  & 1 \\
				& 0.00941     & 191   & 1   & 906  & 1 \\
				& 0.00094     & 18253 & 1   & 822  & 1 \\ \hline
				\multirow{3}{*}{$0.1 \lambda_2^{*}$}
				& 0.04856     & 96091 & 1   & 857  & 1 \\
				& 0.00971     & 507   & 1   & 1056 & 1 \\
				& 0.00097     & 10433 & 1   & 828  & 1 \\ \hline
		\end{tabular}}
	\end{table}

	\subsection{Experiment of Section \ref{sec:sensitivity}} \label{appendix:L0BnB_exps}
	{The number of nodes for all experiments of Section \ref{sec:sensitivity}, and the parameters for varying $n$ and $k$ are presented in Table \ref{table:sensitivity_appendix}. 
		For varying correlation coefficient: for $p=10^4$, we have   $\lambda_0^{*} = 0.0326$, $\lambda_2^{*} = 0.0091$, and $M = 0.465$ (for all values of $\rho)$. For $p=10^5$, we have  $\lambda_0^{*} = 0.0298$, $\lambda_2^{*} = 0.00202$, and $M = 0.449$ for all $\rho$, except for $\rho = 0.9$ which has  $\lambda_0^{*} = 0.0304$.}  
	\begin{table}[htbp]
		\renewcommand{\tabcolsep}{1.2pt}
		\renewcommand{\arraystretch}{1.2}
		\centering
		\caption{{\small{Problem parameters and number of nodes for the experiment of Section \ref{sec:sensitivity}}}.}
		\label{table:sensitivity_appendix}
		{\begin{tabular}{|c|c|c|c|c|}
				\multicolumn{5}{c}{Varying $n$} \\
				\hline
				$\boldsymbol n$ & $10^2$ & $10^3$ & $10^4$ & $10^5$ \\ \hline
				$\boldsymbol \lambda_0$ & 0.0104 & 0.0121 & 0.0152 & 0.0181 \\ \hline
				$\boldsymbol \lambda_2$ & 0.0409 & 0.00139 & 0.00295 & 0.0001 \\ \hline
				$\boldsymbol M$ & 0.410 & 0.337 & 0.329 & 0.316 \\ \hline
				\bf{Nodes} & 1386125 & 383 & 259 & 421 \\ \hline
			\end{tabular}
			\quad
			\begin{tabular}{|c|c|c|c|c|}
				\multicolumn{5}{c}{Varying $k^{\dagger}$} \\
				\hline
				$\boldsymbol k^{\dagger}$ & 10	& 15	& 20	& 25 \\ \hline
				$\boldsymbol \lambda_0$ & 0.0119 & 0.00615 & 0.00415 & 0.00275 \\ \hline
				$\boldsymbol \lambda_2$ & 0.0409 & 0.0409 & 0.0596 & 0.126 \\ \hline
				$\boldsymbol M$ & 0.348 &  0.275 &  0.224 &  0.186 \\ \hline
				\bf{Nodes} & 225 & 5261 & 218667 & 139353 \\ \hline
			\end{tabular}
			\quad
			\newline
			\begin{tabular}{|c|c|c|c|c|c|c|}
				\multicolumn{7}{c}{Varying correlation coefficient  $\rho$} \\
				\hline
				$\boldsymbol \rho$ & 0.1	& 0.3	& 0.5	& 0.7	& 0.8	& 0.9 \\ \hline
				\bf{Nodes} ($p = 10^4$) & 597 & 613 & 845 & 2295 & 7913 & 86083 \\ \hline
				\bf{Nodes} ($p = 10^5$) & 589 & 633 & 943 & 3145 & 11873 & 127267 \\ \hline
		\end{tabular}}
	\end{table}
	
	\subsection{Computing Setup and Software} We used four  machines in our experiments, each using an Intel(R) Xeon(R) Platinum 8252C CPU and $48$GB of RAM. We ran Gurobi and MOSEK using their Python APIs, BARON using the Pyomo Python API \cite{hart2011pyomo}, and \cite{bertsimas2017sparse}'s OA approach using the SubsetSelectionCIO toolkit \cite{bertsimas2017sparseclass} in Julia. Relevant software versions: Ubuntu 18.04.3, Python 3.7.10, R 3.6.3, Julia 0.6.4, Gurobi 9.0.1, MOSEK 9.2.3, BARON 2021.1.13, Pyomo 5.7.1, and L0Learn 2.0. For Julia, exceptionally, we used Gurobi 8.0.1 due to compatibility issues with Gurobi 9.


	%
	%

\end{document}